\documentclass{mcom-l}

\usepackage{algorithm, algorithmic, bm, lmodern, tikz,url}
\usetikzlibrary{plotmarks}
\usetikzlibrary{matrix}

\newtheorem{theorem}{Theorem}[section]
\newtheorem{lemma}[theorem]{Lemma}
\newtheorem{corollary}[theorem]{Corollary}

\theoremstyle{definition}

\newtheorem{example}[theorem]{Example}

\theoremstyle{remark}
\newtheorem{remark}[theorem]{Remark}

\numberwithin{equation}{section}

\newcommand{\abs}[1]{\left\lvert#1\right\rvert}

\newcommand{\norm}[1]{\|#1\|}
\newcommand{\inner}[1]{\langle#1\rangle}
\newcommand{\complex}{\mathbb{C}}
\newcommand{\real}{\mathbb{R}}
\newcommand{\integer}{\mathbb{Z}}

\newcommand{\GL}{\mathrm{GL}}
\newcommand{\PSLQ}{\texttt{PSLQ}}

\begin{document}

\title{The PSLQ Algorithm for Empirical Data}

\author{Yong Feng}
\address{Chongqing Key Lab of Automated Reasoning and Cognition,
	 Chongqing Institute of Green and Intelligent Technology,
	 Chinese Academy of Sciences, Chongqing 400714, China}
\email{yongfeng@cigit.ac.cn}
\thanks{The first author  was supported by NNSF (China) Grant 11671377 and 61572024.}

\author{Jingwei Chen}
\address{Chongqing Key Lab of Automated Reasoning and Cognition,
	Chongqing Institute of Green and Intelligent Technology,
	Chinese Academy of Sciences, Chongqing 400714, China}
\email{chenjingwei@cigit.ac.cn}
\thanks{The second author was supported by NNSF (China) Grant 11501540, CAS ``Light of West China'' Program and Youth Innovation Promotion Association of CAS}
\thanks{The second author is the corresponding author.}

\author{Wenyuan Wu}
\address{Chongqing Key Lab of Automated Reasoning and Cognition,
	Chongqing Institute of Green and Intelligent Technology,
	Chinese Academy of Sciences, Chongqing 400714, China}
\email{wuwenyuan@cigit.ac.cn}
\thanks{The third author was supported by NNSF (China) Grant 11471307 and 11771421, Chongqing Research Program (cstc2015jcyjys40001, KJ1705121), and CAS Research Program of Frontier Sciences (QYZDB-SSW-SYS026).}

\subjclass[2000]{Primary 11A05, 11Y16; Secondary 68-04}

\date{July 17, 2017 and, in revised form, Nov. 8 2017, and Jan. 9, 2018.}


\keywords{Integer relation, PSLQ, empirical data}

\begin{abstract}
The celebrated integer relation finding algorithm PSLQ has been successfully used in many applications. PSLQ was only analyzed theoretically for exact input data, however, when the input data are irrational numbers, they must be approximate ones due to the finite precision of the  computer.  When the algorithm takes empirical data (inexact data with error bounded) instead of exact real numbers as its input,  how do we theoretically ensure  the output of the algorithm  to be an exact integer relation?

In this paper, we investigate the PSLQ algorithm for empirical data as its input. First, we give a termination condition for this case. Secondly, we analyze a perturbation on the hyperplane matrix constructed from the input data and hence disclose a relationship between the accuracy of the input data and the output quality (an upper bound on the absolute value of the inner product of the exact data and the computed integer relation), which naturally leads to an error control strategy for PSLQ. Further, we analyze the complexity bound of the PSLQ algorithm for empirical data. Examples on transcendental numbers and algebraic numbers show the meaningfulness of our error control strategy.
\end{abstract}

\maketitle

\vspace{-.8cm}

\section{Introduction}

A vector $\bm {m}\in\integer^n\setminus\{\bm{0}\}$ is called an \emph{integer relation} for $\bm{\alpha}\in\real^n$ if $\inner{{\bm\alpha}, \bm m} =0$.  The problem of finding integer relations for rational or real numbers can be dated back to the time of Euclid. It is closely related to the problem of finding a small vector in a Euclidean lattice. In fact, the celebrated Lenstra-Lenstra-Lov\'asz (LLL) lattice basis reduction algorithm can be used to find an integer relation. This was already pointed out in \cite[page 525]{LenstraLenstraLovasz1982}. The HJLS algorithm \cite{HastadJustLagariasSchnorr1989} is the first proved polynomial time algorithm for integer relation finding. The PSLQ algorithm \cite{FergusonBailey1992, FergusonBaileyArno1999} is one of the most frequently used algorithms to find integer relations. Both HJLS and PSLQ can be viewed as algorithms to compute the intersection between a lattice and a vector space; see \cite{ChenStehleVillard2013}. For detailed historical notes, we refer to \cite{HastadJustLagariasSchnorr1989, FergusonBaileyArno1999}. Nowadays, integer relation finding has been successfully used in different areas, such as experimental math \cite{BaileyBorweinKimberleyLadd2017,Stenger2017} and physics \cite{BaileyBorwein2015}. For more applications, we refer to \cite{BorweinLisonek2000}  and the references therein.

However, there always exist some data that can only be obtained with limited accuracy. Indeed, all the input data in applications above are of limited accuracy, and hence not exact values. Consequently, it is of great importance to study how to obtain exact integer relations for $\bm \alpha$ from an approximation of $\bar{\bm \alpha}$ by PSLQ.

To the best of our knowledge, there exists only an experimental result on this topic, due to Bailey.  Bailey in \cite{Bailey2000} suggested that if one wishes to recover an integer relation with coefficients bounded by $G$ for an $n$-dimensional vector  $\bm \alpha$, then the input vector $\bm \alpha$ must be specified to at least $n\log_{10} G$ decimal digits, and one must employ floating-point arithmetic with at least $n\log_{10} G$ accurate digits.  Using this experimental result, a lot of non-trivial integer relations have been discovered by several implementations of PSLQ, such as MPFUN90 \cite{Bailey1990}, ARPREC \cite{ARPREC}, etc., all of which employ high precision floating-point  arithmetic. Recently, a PSLQ implementation in a new arbitrary precision package MPFUN2015 \cite{MPFUN2015} has been used  to discover large Poisson polynomials \cite{BaileyBorweinKimberleyLadd2017}, including the largest successful integer relation computations performed to date (using $64,000$ decimal digits), based on the precision estimation suggested by Bailey.
Bailey's precision estimation works well in practice, however lacks theoretical support. In this paper, we attempt to provide a theory for the error control of PSLQ.

Let $\bm \alpha=(\alpha_1, \cdots,\alpha_n)\in\real^n$ be the \textit{intrinsic data} (exact data that may not be known) with an integer relation within a $2$-norm bound $M$, and $\bar{\bm \alpha}$ be the \textit{empirical data} with $\norm{\bm \alpha - \bar{\bm \alpha}}_2<\varepsilon_1$. Generally, $\bar{\bm \alpha}$ may not have an integer relation within the bound $M$. Therefore, the PSLQ algorithm may not terminate when we compute an integer relation from $\bar{\bm \alpha}$ because the element $h_{n,n-1}$ of the hyperplane matrix (see \eqref{Halpha} and Algorithm \ref{algo:PSLQ}) may never be transformed to zero.

So, firstly, we propose a new termination condition for the PSLQ algorithm. Secondly, even if PSLQ returns $\bm m$ from $\bar{\bm \alpha}$, we need to determine whether $\inner{\bm \alpha,\bm m}=0$, without knowing the intrinsic data $\bm{\alpha}$. To do this requires a gap bound $\delta$ for $\abs{\inner{\bm \alpha,\bm m}}$. A so-called \emph{gap bound} for $\abs{\inner{\bm \alpha,\bm m}}$ is that there exist a given $\delta>0$ such that $\abs{\inner{\bm \alpha,\bm m}}>\delta$ whenever $\abs{\inner{\bm \alpha,\bm m}}\ne 0$. If there exists no further information about $\bm \alpha$, then there does not exist a gap bound in general. However, a gap bound can be given when $\alpha_i$'s are algebraic numbers \cite{KannanLenstraLovasz1988, Just1989}. Once we have a gap bound $\delta$ and  $\abs{\inner{\bm \alpha,\bm m}}<\delta$, it guarantees $\inner{\bm\alpha,\bm m}=0$, even without knowing $\bm{\alpha}$. In this paper, we will not discuss the gap bound, but focus on how to estimate $\abs{\inner{\bm \alpha,\bm m}}$ via establishing a relation between $\abs{\inner{\bm \alpha,\bm m}}$ and $\abs{\inner{ \bar{\bm\alpha},\bm m}}$. Thirdly, we analyze the computation complexity of the PSLQ algorithm for empirical data. Finally, we also give some illustrative examples that show how helpful the error control strategies are for applications of PSLQ.

\section[Preliminaries]{Preliminaries}\label{sec:Pre}

For completeness, we recall the PSLQ algorithm in this section. As indicated in \cite{FergusonBaileyArno1999}, PSLQ works for both of the real case and the complex case. For the complex case, it may find a Gaussian integer relation for a given $\bm \alpha\in\complex^n$. For simplicity, we only consider the real case here.

Let $\bm \alpha=(\alpha_1,\cdots,\alpha_n)\in\mathbb{R}^n$ with $\alpha_i\ne 0$ for $i=1,\cdots,n$. Given $\bm \alpha$ as above, define the \emph{hyperplane matrix} $\bm H_\alpha=(h_{i,j})$ with
\begin{equation}\label{Halpha}
h_{i,j} =\left\{
\begin{array}{ll}
  0, & \mbox{if } 1\le i<j\le n-1, \\
  s_{i+1}/s_i, & \mbox{if } 1\le i=j\le n-1, \\
  -\alpha_i\alpha_j/(s_js_{j+1}), &  \mbox{if } 1\le j<i\le n,
\end{array}
\right.
\end{equation}
where $s_j^2=\sum_{k=j}^{n}\alpha_k^2>0$ for $j=1,2,\cdots,n$.

Further, we can assume that $\norm{\bm \alpha}=1$ ($\norm{\cdot}$ is the Euclidean norm), since the hyperplane matrix $\bm H_\alpha$ is scale-invariant with respect to $\bm \alpha$, i.e., $\bm H_\alpha=\bm H_{c\cdot\alpha}$ for $c\in\real\setminus\{0\}$.

Algebraically, PSLQ produces a series of unimodular matrices in $\GL_n(\integer)$ multiplying $\bm H_\alpha$ from left and a series of orthogonal matrices multiplying $\bm H_\alpha$ from right. These matrices are produced by the following subroutines (Algorithm \ref{algo:szr}, \ref{algo:bswp} and \ref{algo:crnr}).

\begin{algorithm}
	\caption{(\texttt{SizeReduction})}\label{algo:szr}
	\begin{algorithmic}[1]
		\REQUIRE A lower trapezoidal $n\times (n-1)$ matrix $\bm{H}=(h_{i,j})$ with $h_{i,j}=0$ if $j>i$ and $h_{j,j}\ne 0$.
		\ENSURE A unimodular matrix $\bm D$ such that $\bm H:=\bm D\cdot \bm H = (h_{i,j})$ satisfying  $|h_{i,j}|\le |h_{j,j}|/2$ for $1\le j<i\le n$.
		
		\STATE $\bm{D} := \bm{I}_n$.
		\FOR {$i$ from $2$ to $n$}
		\FOR {$j$ from $i-1$ to $1$ by stepsize $-1$}
		\STATE {$q :=\lfloor h_{i,j}/h_{j,j}+0.5\rfloor$.}
		\FOR {$k$ from $1$ to $n$}
		\STATE {$d_{i,k}:= d_{i,k}-qd_{j,k}$.}
		\ENDFOR
		\ENDFOR
		\ENDFOR
	\end{algorithmic}
\end{algorithm}

We call the process in Algorithm \ref{algo:szr} \emph{size reduction}. In the PSLQ paper \cite{FergusonBaileyArno1999}, size reduction is called Hermite reduction. To avoid confusedness with the Hermite Normal Form for integral matrices or the Hermite reduction  in the integration of algebraic functions \cite{Hermite1872} (also for creative telescoping) and to be consistent with the similar process used in lattice basis reduction algorithms, we replace ``Hermite reduction'' by ``size reduction''.

\begin{algorithm}[H]
	\caption{(\texttt{BergmanSwap})}\label{algo:bswp}
	\begin{algorithmic}[1]
		\REQUIRE A lower trapezoidal $n\times (n-1)$ matrix $\bm{H}=(h_{i,j})$ with $h_{i,j}=0$ if $j>i$ and $h_{j,j}\ne 0$, and a parameter $\gamma>2/\sqrt{3}$.
		\ENSURE A unimodular matrix $\bm D$ resulting from the exchange of two rows of the identity matrix and the exchange position $r$.
		
		\STATE $\bm{D} := \bm{I}_n$.
		\STATE {Choose $r$ such that $\gamma^r|h_{r,r}|= \max_{j\in\{1,\cdots,n-1\}}\left\{\gamma^j\cdot |h_{j,j}|\right\}$, and then swap the $r$-th row and the $(r+1)$-th row of $\bm D$.}
	\end{algorithmic}
\end{algorithm}

After a Bergman swap, $\bm{H}:=\bm D\bm H$ is usually not  lower trapezoidal. We may multiply the updated $\bm H$ by an orthogonal matrix $\bm{Q}$ from the right such that $\bm{HQ}$ is again a lower trapezoidal matrix. This procedure is called \texttt{Corner}, which is equivalent to performing LQ-decomposition of $\bm H$ (QR-decomposition of $\bm H^T$).
Suppose after a Bergman swap, the $r$-th and $(r+1)$-th rows of $\bm H$ are swapped. Let
\begin{equation}\label{eq:beta}
\eta=h_{r,r},\quad \beta=h_{r+1,r},
\quad \lambda=h_{r+1,r+1}, \quad \delta=\sqrt{\beta^2+\lambda^2}.
\end{equation}
Then we can give the following explicit formula for \texttt{Corner} instead of computing the full LQ-decomposition.

\begin{algorithm}[H]
	\caption{(\texttt{Corner})}\label{algo:crnr}
	\begin{algorithmic}[1]
		\REQUIRE An  $n\times (n-1)$  matrix $\bm H$ that is obtained by a Bergman swap with the $r$-th and $(r+1)$-th rows swapped, where $r<n-1$.
		\ENSURE An orthogonal matrix $\bm Q$ such that $\bm H \bm Q$ is the L-factor of the LQ-decomposition of $\bm H$.
		\STATE Return $\bm Q=(q_{i,j})\in\real^{(n-1)\times(n-1)}$ with
		$$q_{i,j}=\begin{cases}
		\beta/\delta & \text{if $i=r$,$j=r$,}\\
		-\lambda/\delta &\text{if $i=r$,$j=r+1$,}\\
		\lambda/\delta &\text{if $i=r+1$,$j=r$,}\\
		\beta/\delta & \text{if $i=r+1$,$j=r+1$,}\\
		1&\text{$i=j\ne r$ or $i=j\ne r+1$}\\
		0& \text{otherwise}.
		\end{cases}
		$$
	\end{algorithmic}
\end{algorithm}

Now, we are ready to give the following description of the PSLQ algorithm. Note that we suppose that  $\bm \alpha\in\real^n$ has integer relations. In fact, this hypothesis is reasonable, because Babai, Just and Meyer auf der Heide \cite{BabaiJustMeyeraufderHeide1988} showed that under the exact real arithmetic computation model, it is not possible to decide whether there exists a relation for given input $\bm \alpha\in\real^n$. In addition, we omit an early termination condition that checks whether there exists a column of $\bm B$ that is an integer relation, because it does not impact  the  analysis for the worst case.

\begin{algorithm}[H]
	\caption{(\texttt{PSLQ})}
	\label{algo:PSLQ}
	\begin{algorithmic}[1]
		\REQUIRE An $n$-dimensional vector $\bm \alpha=(\alpha_1,\cdots,\alpha_n)$ with $\norm{\alpha}=1$ (suppose that $\bm \alpha$ has integer relations) and  $\gamma>2/\sqrt{3}$.
		\ENSURE An integer relation $\bm m$ for $\bm \alpha$.
		\STATE Construct $\bm{H}_{\alpha}$ as in formula (\ref{Halpha}). Set $\bm{H}:=\bm{H}_{\alpha}$. Set the $n\times n$ matrices $\bm{A}$ and $\bm{B}$ to the identity matrix $\bm{I}_n$. Let $\bm D := \texttt{SizeReduce}(\bm H)$. Update $\bm\alpha:=\bm\alpha \bm D^{-1}$, $\bm{H}:=\bm{DH}$, $\bm{A}:=\bm{DA}$, and $\bm{B}:=\bm{BD^{-1}}$.
		\WHILE{$h_{n,n-1}\neq0$}\label{algstp:term}
		\STATE\label{stp:bswap} Let $(\bm D,\, r) := \texttt{BergmanSwap}(\bm H, \gamma)$, where $\bm D$ is the transform matrix and $r$ is the exchange position. Update $\bm\alpha:=\bm\alpha \bm D^{-1}$, $\bm{H}:=\bm{DH}$, $\bm{A}:=\bm{DA}$, and $\bm{B}:=\bm{BD^{-1}}$.
		\IF {$r<n-1$}
		\STATE\label{stp:corn} Let $\bm Q=\texttt{Corner}(\bm H)$ and update $\bm H:= \bm H\bm Q$.
		\ENDIF
\STATE\label{stp:sr} Let $\bm D := \texttt{SizeReduce}(\bm H)$. Update $\bm\alpha:=\bm\alpha \bm D^{-1}$, $\bm{H}:=\bm{DH}$, $\bm{A}:=\bm{DA}$, and $\bm{B}:=\bm{BD^{-1}}$.
\ENDWHILE
\STATE Return the $(n-1)$-th column of $\bm{B}$.
	\end{algorithmic}
\end{algorithm}

\begin{remark}
	At the beginning, the hyperplane matrix $\bm H_\alpha$ has all diagonal elements nonzero. During the algorithm, all diagonal elements of $\bm H$ are always  nonzero till the termination of $\PSLQ$.
\end{remark}

For the convenient of description, the procedure from step \ref{stp:bswap} to step \ref{stp:sr} in Algorithm \ref{algo:PSLQ} is called an iteration of \PSLQ\, as in \cite[Section 3]{FergusonBaileyArno1999}. 

\begin{theorem}[{\cite[Theorem 2]{FergusonBaileyArno1999}}]
	Assume that $\bm \alpha\in\real^n$ has integer relations. Let $\lambda_\alpha$ be the least $2$-norm of relations for $\bm \alpha$. Then $\PSLQ$ will find an integer relation for $\bm \alpha$ in no more than
	\[
	\binom{n}{2}\frac{\log\left(\gamma^{n-1}\lambda_\alpha)\right)}{\log\tau}
	\]
	iterations, where $\tau = 1/\sqrt{1/\rho^2+1/\gamma^2}$ with $\gamma>2/\sqrt{3}$ and $\rho=2$.
\end{theorem}

\section[PSLQ eps]{The {$\PSLQ_\varepsilon$} Algorithm}\label{sec:invar}

The termination of \texttt{PSLQ} requires one to check whether $h_{n,n-1}=0$. When the input data $\bm \alpha$ with integer relations are exact, it will hold that $h_{n,n-1}=0$ after finitely many iterations of \PSLQ. And hence the output is an integer relation for  $\bm \alpha$.  However, when the input data $\bar{\bm \alpha}$ is an approximation of  $\bm \alpha$, there may not exist any integer relation for  $\bar{\bm \alpha}$. So $h_{n,n-1}$ is usually not equal to zero. This leads to non-termination of $\PSLQ$. Therefore, we first need to explore the termination condition of $\PSLQ$ for empirical data.

\subsection{An Invariant Relation of PSLQ}
Indeed, the quantity $h_{n,n-1}$ plays a very important role in $\PSLQ$, not only for exact data, but also for empirical data. The following theorem gives  a relationship between the $(n-1)$-th column of $\bm B$ ($=\bm A^{-1}$) in \texttt{PSLQ} and $h_{n,n-1}$, which will be shown to be crucial for the study of termination of $\PSLQ$ with empirical data.

Denote by $\bm{H}(k)$ the end result of $\bm{H}$ after exactly $k$ iterations of $\PSLQ$ .

\begin{theorem}\label{lem:zn-1}
	Assume that $\bm{H}(k)=\bm{AH}_{\alpha}\bm{Q}$, where $\bm{H}(k)=(h_{i,j}(k))$ is a lower trapezoidal matrix. Set $(z_1(k),\cdots,z_{n-1}(k),z_{n}(k))=(\alpha_1,\cdots,\alpha_{n-1},\alpha_n)\bm{A}^{-1}$. Then, it holds that
	$$|z_{n-1}(k)|\le \sqrt{\alpha_{n-1}^2+\alpha_n^2}|h_{n,n-1}(k)|. $$
\end{theorem}

\begin{proof}
	From
	$$ (z_1(k),\cdots,z_{n-1}(k),z_n(k))\bm{H}(k) = \bm \alpha\bm{A}^{-1}\bm{A}\bm{H}_\alpha\bm{Q}=\bm \alpha\bm{H}_\alpha\bm{Q}=\bm 0$$
	it follows that
	\[z_{n-1}(k)h_{n-1,n-1}(k)+z_n(k)h_{n,n-1}(k)=0.\]
	From \cite[Lemma 5]{FergusonBaileyArno1999}, it holds that $h_{n-1,n-1}(k)\ne 0$ before termination of Algorithm \ref{algo:PSLQ}.  Then, it is obtained that
	\begin{equation}
	z_{n-1}(k)=-\frac{z_n(k)}{h_{n-1,n-1}(k)}h_{n,n-1}(k). \label{equ:control}
	\end{equation}
	
	We claim that $|\frac{z_n(k)}{h_{n-1,n-1}(k)}|$ does not increase as $k$ increases.
In Algorithm \ref{algo:PSLQ}, this quantity can be possibly changed only in \texttt{SizeReduce}, \texttt{BergmanSwap} and \texttt{Corner}, so we next consider them respectively.
When the size reduction (step \ref{stp:sr}) is performed on row $i\le n-1$ of $\bm{H}$, $z_n$ and $h_{n-1,n-1}$ are unchanged, so $|\frac{z_n}{h_{n-1,n-1}}|$ is unchanged. When $i=n$, the size reduction matrix is as follows
	$$\bm D=\begin{pmatrix}
	1&0&0&\cdots&0&0\\
	0&1&0&\cdots&0&0\\
	0&0&1&\cdots &0&0\\
	\vdots &\vdots &\vdots&\vdots&\vdots&\vdots\\
	0&0&0&\cdots &1&0\\
	k_1&k_2&k_3&\cdots& k_{n-1}&1
	\end{pmatrix}=\begin{pmatrix}
	\bm{I}_{n-1}&0\\
	\bm K&1
	\end{pmatrix},
	$$
	where $\bm K=(k_1,\cdots,k_{n-1})$ is an integer vector, and $\bm{I}_{n-1}$ is the $(n-1)\times (n-1)$ identify matrix. Its inverse is
	$$D^{-1}=\begin{pmatrix}
	1&0&0&\cdots&0&0\\
	0&1&0&\cdots&0&0\\
	0&0&1&\cdots &0&0\\
	\vdots &\vdots &\vdots&\vdots&\vdots&\vdots\\
	0&0&0&\cdots &1&0\\
	-k_1&-k_2&-k_3&\cdots& -k_{n-1}&1
	\end{pmatrix}=\begin{pmatrix}
	\bm{I}_{n-1}&0\\
	-\bm K&1
	\end{pmatrix}. $$
	It is easy to see that the $n$-th column of $\bm{A}^{-1}\bm{D}^{-1}$ is the same as that of $\bm{A}^{-1}$. Therefore, $z_n$ is unchanged. On the other hand, $h_{n-1,n-1}$ is also unchanged after size reduction. Hence, $\frac{z_n}{h_{n-1.n-1}}$ is unchanged.
	In step \ref{stp:bswap} of Algorithm \ref{algo:PSLQ},  the Bergman swap is performed between  the $r$-th and $(r+1)$-th rows. When $r<n-2$, it is obvious that $z_n$ and $h_{n-1,n-1}$ are unchanged.  When $r=n-2$, the columns $n-2$ and $n-1$ of $\bm{A}^{-1}$ are swapped. So the $n$-th column of $\bm{A}^{-1}$ is unchanged and $z_n$ is also unchanged,that is $z_n(k+1)=z_n(k)$, while $h_{n-1,n-1}$ is changed as follows. Before step \ref{stp:bswap}, let $\eta=h_{n-2,n-2}(k)$, $\beta=h_{n-1,n-2}(k)$, $\lambda=h_{n-1,n-1}(k)$ and $\delta=\sqrt{\beta^2+\lambda^2}$, then we have
	$$\begin{pmatrix}
	\eta &0\\
	\beta &\lambda
	\end{pmatrix}\xrightarrow{\text{step \ref{stp:bswap}}}\begin{pmatrix}
	\beta&\lambda\\
	\eta&0
	\end{pmatrix}\xrightarrow{\text{step \ref{stp:corn}}}\begin{pmatrix}
	\delta &0\\
	\frac{\eta\beta}{\delta}&-\frac{\eta\lambda}{\delta}
	\end{pmatrix}.
	$$
	Therefore, after step \ref{stp:corn}, the new $h_{n-1,n-1}(k+1)=-\frac{\eta\lambda}{\delta}$. Since the swap occurs at rows $n-2$ and $n-1$, it holds that $|\eta| >\gamma|\lambda|$. Note that $|\beta|<\frac{|\eta|}{\rho}$ yields
	$$\left|\frac{-\eta}{\delta}\right|=\frac{1}{\sqrt{\frac{\beta^2}{\eta^2}+\frac{\lambda^2}{\eta^2}}}
	>\frac{1}{\sqrt{\frac{1}{\rho^2}+\frac{1}{\gamma^2}}}=\tau,$$
where $\rho=2$. So, it follows that
$$|h_{n-1,n-1}(k+1)|=|-\frac{\eta\lambda}{\delta}|>\tau|\lambda|.$$
Hence, it holds that
$$\left|\frac{z_n(k+1)}{h_{n-1,n-1}(k+1)}\right|<\frac{z_n(k)}{\lambda}\frac{1}{\tau}
	=\frac{1}{\tau}\abs{\frac{z_n(k)}{h_{n-1,n-1}(k)}}.$$
	Since $\frac{1}{\tau}<1$, it implies that $\left|\frac{z_n}{h_{n-1,n-1}}\right|$  decreases.
	When $r=n-1$, rows $n-1$ and $n$ of $\bm{H}$ are swapped, so are columns  $n-1$ and $n$ of $\bm{A}^{-1}$. Hence $h_{n-1,n-1}$ and $h_{n,n-1}$ are swapped, and $z_{n-1}$ and $z_n$ are exchanged. Therefore, $h_{n-1,n-1}(k+1)=h_{n,n-1}(k)$ and $z_n(k+1)=z_{n-1}(k)$. From formula (\ref{equ:control}), it follows that
	\begin{equation*}
	z_n(k+1)=z_{n-1}(k)=-\frac{h_{n,n-1}(k)}{h_{n-1,n-1}(k)}z_n(k)=-\frac{h_{n-1,n-1}(k+1)}{h_{n-1,n-1}(k)}z_n(k).
	\end{equation*}
	In this case, $\left|\frac{z_n}{h_{n-1,n-1}}\right|$ remains unchanged.
	Up to now, we have shown that $|\frac{z_n}{h_{n-1,n-1}}|$ either decreases or remains unchanged after the $(k+1)$-th iteration of $\PSLQ$. At the beginning of $\PSLQ$, we have that $z_n(1)=\alpha_n$ and $h_{n-1,n-1}(k)=\frac{|\alpha_n|}{\sqrt{\alpha_{n-1}^2+\alpha_n^2}}$. Hence
	$$\left|\frac{z_n(k)}{h_{n-1,n-1}(k)}\right|\le
	\left|\frac{z_n(1)}{h_{n-1,n-1}(1)}\right|\le\sqrt{\alpha_{n-1}^2+\alpha_n^2},$$
	which completes the proof.
\end{proof}

The property presented in Theorem \ref{lem:zn-1} is an invariant of \texttt{PSLQ} in the sense that it always holds during the algorithm. Furthermore, Theorem \ref{lem:zn-1}  can be used to design an algorithm to find approximate integer relations in the following sense. Given $\bm\alpha$ which may not have an integer relation, if we take the $(n-1)$-th column of $\bm B$ as an approximate integer relation for $\bm\alpha$ in algorithm \ref{algo:PSLQ},  Theorem \ref{lem:zn-1} gives an error estimate, i.e., if \texttt{PSLQ} returns the $(n-1)$-th column of $\bm B$, denoted by $\bm m$, when  $\abs{h_{n,n-1}}< \varepsilon_2$, then
\[\abs{\inner{\bm \alpha,\bm m}}\le\sqrt{\alpha_{n-1}^2+\alpha_n^2}\,\varepsilon_2.\]
Now we improve the algorithm as follows.

\begin{algorithm}[H]
	\caption{(\texttt{PSLQ$_{\varepsilon}$})}
	\label{algo:nPSLQ}
	\begin{algorithmic}[1]
		\REQUIRE A lower trapezoidal matrix $\bm H\in\real^{n\times(n-1)}$ with all diagonal entries nonzero,  $\varepsilon_2>0$ and  $\gamma>2/\sqrt{3}$.
		
		

		\ENSURE An $n$-dimensional integer vector $\bm m$.
		\STATE Set the $n\times n$ matrices $\bm{A}$ and $\bm{B}$ to the identity matrix $\bm{I}_n$.
		 Let $\bm D := \texttt{SizeReduce}(\bm H)$. Update $\bm{H}:=\bm{DH}$, $\bm{A}:=\bm{DA}$, and $\bm{B}:=\bm{BD^{-1}}$.
		\WHILE{ $\abs{h_{n,n-1}}\ge\varepsilon_2$}\label{algstep:term}
		\STATE Let $(\bm D,\, r) := \texttt{BergmanSwap}(\bm H, \gamma)$, where $\bm D$ is the transform matrix and $r$ is the exchange position. Update $\bm{H}:=\bm{DH}$, $\bm{A}:=\bm{DA}$, and $\bm{B}:=\bm{BD^{-1}}$.
		\IF {$r<n-1$}
		\STATE Let $\bm Q=\texttt{Corner}(H)$ and update $\bm H:= \bm H\bm Q$.
		\ENDIF
\STATE Let $\bm D := \texttt{SizeReduce}(\bm H)$. Update $\bm{H}:=\bm{DH}$, $\bm{A}:=\bm{DA}$, and $\bm{B}:=\bm{BD^{-1}}$.
		\ENDWHILE
\STATE Return the $(n-1)$-th column of $\bm{B}$.
	\end{algorithmic}
\end{algorithm}

Besides the termination condition being replaced by $\abs{h_{n,n-1}}<\varepsilon_2$, the main difference of $\PSLQ_{\varepsilon}$ from  $\PSLQ$ is that the input is changed as a more general lower trapezoidal matrix which may not satisfy the fine structure in \eqref{Halpha}. The remainder of this section will be devoted to analyze $\PSLQ_{\varepsilon}$.

\subsection{Termination and Complexity}\label{sec:term}
We now show that $\PSLQ_{\varepsilon}$ terminates after finitely many number of iterations  stated in the following theorem.

\begin{theorem}\label{thm:term}
	Given $\bm{H}\in\real^{n\times(n-1)}$, $\PSLQ_{\varepsilon}$ terminates in no more than
	\[
	\frac{n(n+1)((n-1)\log\gamma + \log\frac{1}{\varepsilon_2})}{2\log\tau}
	\]
	iterations, where $\tau = 1/\sqrt{1/\rho^2+1/\gamma^2}$.
\end{theorem}

\begin{proof}
	Define the $\Pi$ function after $k$ iterations as follows
	\[
	\Pi(k) = \prod_{j=1}^{n-1}\max\left(\abs{h_{i,i}(k)}, \frac{h_{\max}(k)}{\gamma^{n-1}}\right)^{n-j},
	\]
	where $h_{\max}(k)$ is the maximum of $\abs{h_{i,i}(k)}$ for $i=1,2,\cdots, n-1$. Then the proof is similar to the proof of \cite[Theorem 2]{FergusonBaileyArno1999}. See Appendix \ref{app:term} for the full proof.
\end{proof}

Note that if $\bm H$ is the hyperplane matrix for  $\bm\alpha\in\real^n$ and $\bm \alpha$ has an integer relation, let $M_{\alpha}$ be the minimal $2$-norm of integer relations for $\bm\alpha$. Then from  \cite[Theorem 1]{FergusonBaileyArno1999}, it holds that $\frac{1}{h_{\max}(k)}\le M_{\alpha}$. From inequality \eqref{eq:alg2}, it can be obtained that
\[
k\le \frac{n(n-1)((n-1)\log\gamma + \log\frac{1}{h_{\max}(k)})}{2\log\tau}\le \frac{n(n-1)((n-1)\log\gamma + \log M_{\alpha})}{2\log\tau},
\]
which is the same as \cite[Theorem 2]{FergusonBaileyArno1999}.

\subsection[Perturbation Analysis]{Perturbation Analysis of $\PSLQ_{\varepsilon}$}\label{sec:perturb}

Before we present the technical details, we recall some notations. For the intrinsic data $\bm \alpha$, we assume that we can only obtain  the corresponding empirical data $\bar{\bm \alpha}$ with $\norm{\bm \alpha - \bar{\bm \alpha}}_2<\varepsilon_1$. For $\bar{\bm \alpha}$, we can construct its hyperplane matrix ${\bm H_{\bar{\alpha}}}$ as in \eqref{Halpha}. But we can not use ${\bm H_{\bar{\alpha}}}$ as the input matrix for $\PSLQ_{\varepsilon}$. Instead, we use $\overline{\bm{H}}_\alpha$ to represent a more general perturbation to $\bm H_\alpha$ including round-off errors in computing $\bm H_{\bar{\bm \alpha}}$, which only keeps the lower trapezoidal structure and satifies
\begin{equation}\label{eq:e3}
\|\overline{\bm{H}}_\alpha-\bm{H}_\alpha\|_F\le \varepsilon_3,
\end{equation}
where $\norm{\cdot}_F$ is the matrix Frobenius norm.
Suppose that one wants to find an integer relation for $\bm \alpha\in\real^n$ by using $\PSLQ_\varepsilon$, where the input is $\overline{\bm{H}}_\alpha$, the termination condition is $\abs{h_{n,n-1}}< \varepsilon_2$ and the output is $\bm m$.
Next, we  investigate the relations among $\varepsilon_2$, $\varepsilon_3$ and $\abs{\inner{\bm m,\bm\alpha}}$.

Denote by $\bm H_{[1..n-1]}$  the submatrix of $\bm H_\alpha$ that consists of the first $n-1$ rows and the first $n-1$ columns.  It follows from \eqref{eq:e3} that $\|\overline{\bm{H}}_{[1..n-1]}-\bm{H}_{[1..n-1]}\|_F\le \varepsilon_3$.

First, we give  explicit formulae for the F-norm of $\bm{H}_{[1..n-1]}$ and $\bm{H}_{[1..n-1]}^{-1}$; see Appendix \ref{sec:condnum} for the proof.

\begin{lemma}\label{lem:normH}
	Let the notations be as above. Then
	\[
	\begin{split}
	\|\bm{H}_{[1..n-1]}^{-1}\|_F^2&=(n-2)+\frac{\|\bm \alpha\|^2}{\alpha_n^2},\\
	\|\bm{H}_{[1..n-1]}\|_F^2&=(n-2)+\frac{\alpha_n^2}{\|\bm \alpha\|^2}.
	\end{split}
	\]
\end{lemma}

The following lemma enables us to give an estimation on
$\norm{\overline{\bm{H}}_{[1..n-1]}^{-1}}_F$.

\begin{lemma}[{\cite[Theorem 2.3.4]{GolubvanLoan2013}}]\label{lem:matnorm}
	Let $\bm{A}$ be a nonsingular matrix with  perturbation $\bm{E}$. Let $\|.\|$ denote any matrix norm satisfying inequality $\|\bm B\bm C\|\le \|\bm B\|\|\bm C\|$ for any matrices $\bm B$ and $\bm C$. If  $\|\bm{EA}^{-1}\|<1$, then $\bm{A}+\bm{E}$ is nonsingular, and it holds that 
	$$\|(\bm{A}+\bm{E})^{-1}-\bm{A}^{-1}\|\le \frac{\|\bm{EA}^{-1}\|}{1-\|\bm{EA}^{-1}\|}\|\bm{A}^{-1}\|.$$
\end{lemma}

Applying the above lemma to $\bm{H}_{[1..n-1]}$ yields the following corollary.

\begin{corollary}\label{cor:est_H}
	Let $\overline{\bm{H}}_\alpha=\bm{H}_\alpha+\Delta \bm{H}_\alpha$ and $\|\Delta\bm{H}_\alpha\|_F<\varepsilon_3$,and let $\bm{H}_{[1..n-1]}$ and $\overline{\bm{H}}_{[1..n-1]}$ denote  submatrices  consisting of the first $(n-1)$ rows and the first $(n-1)$ columns of $\bm{H}_\alpha$ and $\overline{\bm{H}}_\alpha$ respectively. When $\varepsilon_3<\frac{1}{\|\bm{H}_{[1..n-1]}^{-1}\|_F}$, $\overline{\bm{H}}_{[1..n-1]}$ is nonsingular and it holds that
	$$\|\overline{\bm{H}}^{-1}_{[1..n-1]}\|_F\le \frac{1}{1-\varepsilon_3 \|\bm{H}_{[1..n-1]}^{-1}\|_F}\|\bm{H}_{[1..n-1]}^{-1}\|_F.$$
\end{corollary}

\begin{proof}
	When $\varepsilon_3<\frac{1}{\|\bm{H}_{[1..n-1]}^{-1}\|_F}$, it holds that
	\[
	\begin{split}
	\|\Delta\bm{H}_{[1..n-1]}\bm{H}_{[1..n-1]}^{-1}\|_F&\le
	\|\Delta\bm{H}_{[1..n-1]}\|_F\|\bm{H}_{[1..n-1]}^{-1}\|_F\\
	&\le \|\Delta\bm{H}_\alpha\|_F\cdot\|\bm{H}_{[1..n-1]}^{-1}\|_F\\&
	<\varepsilon_3\cdot \|\bm{H}_{[1..n-1]}^{-1}\|_F<1.
	\end{split}
	\]
	From Lemma \ref{lem:matnorm}, $\overline{\bm{H}}_{[1..n-1]}$ is nonsingular and
	it follows that
	\[
	\begin{split}
	\|\overline{\bm{H}}_{[1..n-1]}^{-1}\|_F&<\frac{1}{1-\|\Delta\bm{H}_{[1..n-1]}
		\bm {H}_{[1..n-1]}^{-1}\|_F}\|\bm{H}_{[1..n-1]}^{-1}\|_F\\
	&\le\frac{1}{1-\varepsilon_3\|\bm{H}_{[1..n-1]}^{-1}\|_F}\|\bm{H}_{[1..n-1]}^{-1
	} \|_F .
	\end{split}
	\]
	This completes the proof.
\end{proof}

Corollary  \ref{cor:est_H} shows that when
$\varepsilon_3<{1}/{\|\bm{H}_{[1..n-1]}^{-1}\|_F}$, it holds that
$\overline{h}_{i,i}\ne 0$ for $i=1,\cdots,n-1$. Denote by $\overline{\bm
	\alpha}=(\overline{\alpha}_1,\cdots,\overline{\alpha}_n)$ a unit real vector
satisfying $\overline{\bm \alpha}\overline{\bm{H}}_{\alpha}=0$. Without loss of generality, we assume that
$\overline{\alpha}_n\ne 0$. Otherwise we can deduce $\overline{\bm \alpha}=\bm
0$, which contradicts to that $\overline{\bm
	\alpha}$ is a unit vector. (In fact, since $\bar{\alpha}_{n-1}\bar{h}_{n-1,n-1}+\bar{\alpha}_{n}\bar{h}_{n,n-1}=0$ and   $\overline{h}_{n-1,n-1}\ne 0$ we have $\bar{\alpha}_{n}=0$ implies $\bar{\alpha}_{n-1}=0$. Similarly, $\bar{\alpha}_{i}=0$ for $i=1,2,\cdots n-2$.) Moreover, we can choose vector $\overline{\bm\alpha}$ with
$\overline{\alpha}_n>0$. Next, we give a nonzero lower bound on
$\overline{\alpha}_n$.

\begin{lemma}\label{lemma:method_proof}
Let  $\bm \xi=(\xi_1,\cdots,\xi_{n-1},1)$ be a real vector with $\|\bm \xi\|\le M$ and let  $\bm\beta=\frac{\bm\xi}{\|\bm\xi\|}=(\beta_1,\cdots,\beta_{n-1},\beta_n)$. Then it holds that
	$\abs{\beta_n}\ge \frac{1}{M}$.
\end{lemma}

\begin{proof}
	According to assumptions,
	$$1=\|\bm\beta\|=\abs{\beta_n\left(\frac{\beta_1}{\beta_n},\cdots,
		\frac{\beta_{n-1}}{\beta_n},1\right)}=\|\beta_n\bm\xi\|\le |\beta_n|\cdot\|\bm\xi\|\le |\beta_n|M.$$
	The proof of lemma is finished.
\end{proof}

The above lemma enables us to give a lower bound of some component of a unit vector.
\begin{lemma}\label{lemma:est_low}
	Let $\overline{\bm \alpha}=(\overline{\alpha}_1,\cdots,\overline{\alpha}_{n-1},\overline{\alpha}_n)$ be a unit vector such that $\overline{\bm \alpha}\overline{\bm{H}}_\alpha=0$.  If $\varepsilon_3$ given in \eqref{eq:e3} is less than $\frac{\alpha_n}{\sqrt{(n-2)\alpha_n^2+1}}$, then
	$$|\overline{\alpha}_n|\ge \frac{\alpha_n}{2\sqrt{1-\alpha_n^2}\sqrt{(n-2)\alpha_n^2+1}+2\alpha_n}.$$
\end{lemma}

\begin{proof}
	Consider the linear system $(x_1,\cdots,x_n)\overline{\bm{H}}_\alpha=0$ with unknowns $x_i$  for $i=1,2,\cdots,n$. Since the rank of  $\overline{\bm{H}}_\alpha$ is at most $n-1$, we can assume that $x_n=1$; then it reduces to the following system:
	$$(x_1,\cdots,x_{n-1})\overline{\bm{H}}_{[1..n-1]}=-(\overline{h}_{n,1},\cdots,\overline{h}_{n,n-1}).$$
	If $\varepsilon_3<\frac{\alpha_n}{2\sqrt{(n-2)\alpha_n^2+1}}$, then $\overline{\bm{H}}_{[1..n-1]}$ is nonsingular by Lemma \ref{lem:normH} and Corollary \ref{cor:est_H}, so $(x_1,\cdots,x_{n-1})=-(\overline{h}_{n,1},\cdots,\overline{h}_{n,n-1})\overline{\bm{H}}_{[1..n-1]}^{-1}$. Hence, it holds that
	\begin{equation*}
	\begin{split}
	\|(x_1,\cdots,x_{n-1})\|_2&\le \|(\overline{h}_{n,1},\cdots,\overline{h}_{n,n-1})\|_2\|\overline{\bm{H}}_{[1..n-1]}^{-1}\|_2\\
	&\le \|(\overline{h}_{n,1},\cdots,\overline{h}_{n,n-1})\|_2\|\overline{\bm{H}}_{[1..n-1]}^{-1}\|_F \\
	&\le (\|(h_{n,1},\cdots,h_{n,n-1})\|_2+\varepsilon_3)\|\overline{\bm{H}}_{[1..n-1]}^{-1}\|_F\\
	&\le \left(\sqrt{1-\alpha_n^2}+\frac{\alpha_n}{2\sqrt{(n-2)\alpha_n^2+1}}\right)2\|\bm{H}_{[1..n-1]}^{-1}\|_F\\
	&=2\sqrt{1-\alpha_n^2}\|\bm{H}_{[1..n-1]}^{-1}\|_F+1\\
	&=\frac{2\sqrt{1-\alpha_n^2}\sqrt{(n-2)\alpha_n^2+1}}{\alpha_n}+1.
	\end{split}
	\end{equation*}
	Thus, it is obtained that
	\begin{equation*}
	\begin{split}
	\|(x_1,\cdots,x_{n-1},x_n)\|_2&\le \|(x_1,\cdots,x_{n-1})\|_2+1\\
	&\le \frac{2\sqrt{1-\alpha_n^2}\sqrt{(n-2)\alpha_n^2+1}}{\alpha_n}+2\\
	&=\frac{2\sqrt{1-\alpha_n^2}\sqrt{(n-2)\alpha_n^2+1}+2\alpha_n}{\alpha_n}.
	\end{split}
	\end{equation*}
	From Lemma \ref{lemma:method_proof}, it follows that
	$$ |\overline{\alpha}_n|\ge\frac{1}{\norm{(x_1,x_2,\cdots,x_n)}}\ge\frac{\alpha_n}{2\sqrt{1-\alpha_n^2}\sqrt{(n-2)\alpha_n^2+1}+2\alpha_n}.$$
	The proof of the lemma is finished.
\end{proof}

We now give the main theorem of this paper, which can be seen as a forward error analysis of $\PSLQ_{\varepsilon}$ for the perturbation introduced in \eqref{eq:e3}.

\begin{theorem}\label{thm:forward}
	Given a real vector $\bm \alpha=(\alpha_1,\cdots,\alpha_n)$, let
	$\bm{H}_\alpha$ be the hyperplane matrix constructed as in (\ref{Halpha}).
	Let
	$\overline{\bm{H}}_\alpha$ be an approximate matrix of $\bm{H}_\alpha$ with	$\|\bm{H}_\alpha-\overline{\bm{H}}_\alpha\|_F<\varepsilon_3<\frac{\alpha_n}{2\sqrt{(n-2)\alpha_n^2+1}}$. Let $\bm{A}$ be the unimodular matrix  and let $\bm{Q}$ be the orthogonal matrix such that $\bm{H}=(h_{i,j})=\bm{A\overline{H}}_\alpha\bm{Q}$ is a lower
	trapezoidal matrix at the termination of $\PSLQ_\varepsilon$ with $|h_{n,n-1}|<\varepsilon_2$. Let
	$\bm m$ denote the $(n-1)$-th column of $\bm{A}^{-1}$. Then
	$$|\inner{\bm \alpha,\bm m}|<C\cdot(\|\bm m\|\varepsilon_3+\alpha_n\varepsilon_2) ,$$
	where $C = \frac{2(\sqrt{(n-2)\alpha_n^2+1}+\alpha_n)}{\alpha_n}$.
\end{theorem}

\begin{proof}
	Suppose that $\PSLQ_\varepsilon$ returns $\bm m$ with $\overline{\bm{H}}_{\alpha}$ as the hyperplane matrix, when  $\abs{h_{n,n-1}}< \varepsilon_2$. Then this process can be seen as running $\PSLQ_\varepsilon$ for a unit vector $\overline{\bm \alpha}$ satisfying  $\overline{\bm \alpha}\overline{\bm{H}}_\alpha=0$. According to Theorem \ref{lem:zn-1}, we have  $\abs{\inner{\overline{\bm \alpha},\bm m}}\le \varepsilon_2$.
	Now we consider the  following system
	\begin{equation}\label{equ:c}
	\overline{\bm{H}}_{\alpha}
	\bm c
	=\bm m+ (0,0,\cdots,0,b)^T
	\end{equation}
	where $\bm c = (c_1,c_2,\cdots,c_{n-1})^T$ is the unknown vector. We have that
	$$0=\overline{\bm \alpha}\overline{\bm{H}}_{\alpha}\bm c =\inner{\overline{\bm \alpha},\bm m}+\overline{\bm \alpha}_nb.$$
	Hence $\overline{\alpha}_nb=-\inner{\overline{\bm \alpha},\bm m}$, and we have
	\[|b|<\frac{\varepsilon_2}{\overline{\alpha}_n}.\]
From \eqref{equ:c} it holds that
	\begin{eqnarray*}
		|\inner{\bm \alpha,\bm m}|&=&\left|\bm \alpha\overline{\bm{H}}_{\alpha}\bm c-\alpha_n b\right|\le \left|\bm \alpha\overline{\bm{H}}_{\alpha}\bm c\right|+|\alpha_n|| b| \\
		&\le&\left|\bm \alpha(\overline{\bm{H}}_{\alpha}-\bm{H}_{\alpha})\bm c\right|+|\alpha_n|| b|  \\
		&\le& \|\bm \alpha\|\|\overline{\bm{H}}_{\alpha}-\bm{H}_{\alpha}\|_2\left\|\bm c\right\|+|\alpha_n|| b| \\
		&\le& \|\bm  \alpha\|\|\overline{\bm{H}}_{\alpha}-\bm{H}_{\alpha}\|_2\left\|\bm c\right\|+|\alpha_n|| b| \\
		&\le&\|\bm \alpha\|\left\| \bm c\right\|\varepsilon_3+\frac{|\alpha_n|}{|\overline{\alpha}_n|}\varepsilon_2.
	\end{eqnarray*}
	Since  $\varepsilon_3<\frac{\alpha_n}{2\sqrt{(n-2)\alpha_n^2+1}}$,
	by Lemma \ref{lemma:est_low}, it follows that
	\begin{equation}\label{equ:est_tem}
	|\inner{\bm \alpha,\bm m}|<\|\bm \alpha\|\cdot\left\| \bm c\right\|\varepsilon_3+2(\sqrt{1-\alpha_n^2}\sqrt{(n-2)\alpha_n^2+1}+\alpha_n)\varepsilon_2.
	\end{equation}
	The first $n-1$ equations of \eqref{equ:c} give a square system
	$$\overline{\bm{H}}_{[1..n-1]}\begin{pmatrix}
	c_1\\
	\vdots\\
	c_{n-1}
	\end{pmatrix}
	=\begin{pmatrix}
	m_1\\
	\vdots\\
	m_{n-1}
	\end{pmatrix}.$$
	Then it is obtained that
	\begin{equation*}
	\left\|\bm c\right\|\le \|\overline{\bm{H}}_{[1..n-1]}^{-1}\|_2\left\|\begin{pmatrix}
	m_1\\
	\vdots\\
	m_{n-1}
	\end{pmatrix}\right\|_2\le \|\overline{\bm{H}}_{[1..n-1]}^{-1}\|_F\|\bm m\|.
	\end{equation*}
	Since $\varepsilon_3<\frac{\alpha_n}{2\sqrt{(n-2)\alpha_n^2+1}}$, by Corollary \ref{cor:est_H} we have
	\begin{equation*}
	\begin{split}
	\left\|\bm c\right\|&\le \frac{1}{1-\varepsilon_3 \|\bm{H}_{[1..n-1]}^{-1}\|_F}\|\bm{H}_{[1..n-1]}^{-1}\|_F\|\bm m\|_2\\&<2\|\bm{H}_{[1..n-1]}^{-1}\|_F\|\bm m\|_2<\frac{2\sqrt{(n-2)\alpha_n^2+1}}{\alpha_n}\|\bm m\|_2.
	\end{split}
	\end{equation*}
	Substituting the above inequality into  (\ref{equ:est_tem}) yields
	\[\begin{split}
	|\inner{\bm \alpha,\bm m}|&<\|\bm \alpha\|\left\| \bm c\right\|\varepsilon_3+2(\sqrt{1-\alpha_n^2}\sqrt{(n-2)\alpha_n^2+1}+\alpha_n)\varepsilon_2\\
	&<\frac{2\sqrt{(n-2)\alpha_n^2+1}}{\alpha_n}\|\bm m\|\varepsilon_3+2(\sqrt{1-\alpha_n^2}\sqrt{(n-2)\alpha_n^2+1}+\alpha_n)\varepsilon_2\\
	&<\frac{2(\sqrt{(n-2)\alpha_n^2+1}+\alpha_n)}{\alpha_n}\|\bm m\|\varepsilon_3+2(\sqrt{(n-2)\alpha_n^2+1}+\alpha_n)\varepsilon_2\\
	&=\frac{2(\sqrt{(n-2)\alpha_n^2+1}+\alpha_n)}{\alpha_n}(\|\bm m\|\varepsilon_3+\alpha_n\varepsilon_2).
	\end{split}
	\]
	The theorem is proved.
\end{proof}

Although the quantity $\abs{\inner{\bm \alpha,\bm m}}$  is usually nonzero for empirical data, it  measures somewhat how close is from $\bm m$ to  a true integer relation for $\bm \alpha$. So it  can be seen as output error. In this sense,
Theorem \ref{thm:forward} says that if a perturbation of the input $\bm{H}_\alpha$ is small enough then the ``output error'' of $\PSLQ_{\varepsilon}$ can  be also small. Roughly speaking, if we fix the termination condition $\varepsilon_2$ to be a tiny number, then the ``output error'' is amplified by a factor at most $C\cdot\|\bm m\|$.

\section[PSLQ with Empirical Data]{$\PSLQ_\varepsilon$ with Empirical Data}\label{sec:emp}

Aiming to obtain $\bm m$ by $\PSLQ_\varepsilon$ such that $\abs{\inner{\bm{\alpha},\bm{m}}}<\varepsilon$, we study how to determine the error control parameters $\varepsilon_1$, $\varepsilon_2$ and $\varepsilon_3$ in this section.

\subsection[Error Control of PSLQ]{Error Control of {$\PSLQ_{\varepsilon}$}}

\begin{lemma}\label{lem:Hbar}
	Let $\bm \alpha=(\alpha_1,\cdots,\alpha_n)$ be an $n$-dimensional unit vector with $|\alpha_n|=\max_i\{|\alpha_i|\}$ and let $\bar{\bm \alpha}$ be its approximation. Construct $\bm{H}_\alpha$ and $\bm{H}_{\bar{\alpha}}$ as in (\ref{Halpha}) for $\bm \alpha$ and $\bar{\bm \alpha}$ respectively. If $\|\bm \alpha-\bar{\bm\alpha}\|<\frac{1}{8n}$, Then it holds that
	\[
	\|\bm{ H}_\alpha - \bm{ H}_{\bar{\alpha}}\|_F< 8n^{\frac{3}{2}}\|\bm{\alpha} - \overline{\bm{\alpha}}\|.
	\]
\end{lemma}
\begin{proof}
	Let $s_i=\sqrt{\sum_{k=i}^n\alpha_k^2}$, let $\bar{s}_i=\sqrt{\sum_{k=i}^n\bar{\alpha}_k^2}$, let $\bm b_i=(0,\cdots,0,\alpha_i,\cdots,\alpha_n)$ and let $\bar{\bm b}_i=(0,\cdots,0,\bar{\alpha}_i,\cdots,\bar{\alpha}_n)$. It obviously holds that $\|\bm b_i-\bar{\bm b}_i\|\le \|\bm \alpha-\bar{\bm \alpha}\|$. So, it is obtained that
	$|s_i-\bar{s}_i|=|\|\bm b_i\|-\|\bar{\bm b}_i\||\le\|\bm b_i-\bar{\bm b}_i\|\le \|\bm \alpha-\bar{\bm \alpha}\|$. By the way, from $|\alpha_n|=\max_i\{|\alpha_i|\}$ and  $\norm{\bm \alpha}=1$, it follows that $|\alpha_n|\ge \frac{1}{\sqrt{n}}$. Thus, if $\|\bm \alpha-\bar{\bm\alpha}\|<\frac{1}{2\sqrt{n}}$, then it holds that $|\bar\alpha_n|>\frac{1}{2\sqrt{n}}$.
	
	Recall $\bm{H}_{\alpha}=(h_{i,j})$ and
	$$h_{i,j}=\begin{cases}
	\frac{s_{i+1}}{s_i} & \text{If $i=j$}\\
	-\frac{\alpha_i\alpha_j}{s_js_{j+1}} & \text{else if $i>j$ }\\
	0 &  \text{otherwise}.
	\end{cases} $$
	Let us consider the error of $\frac{s_{i+1}}{s_i}$:
	\[
	\begin{split}	\left|\frac{s_{i+1}}{s_i}-\frac{\bar{s}_{i+1}}{\bar{s}_i}\right|&=\left|\frac{s_{i+1}\bar{s}_i-s_i\bar{s}_{i+1}}{s_i\bar{s}_i}\right|
	=\left|\frac{s_{i+1}\bar{s}_i-s_is_{i+1}+s_is_{i+1}-s_i\bar{s}_{i+1}}{s_i\bar{s}_i}\right|\\
	&\le \frac{s_{i+1}|s_i-\bar{s}_i|}{s_i\bar{s}_i}+\frac{s_i|s_{i+1}-\bar{s}_{i+1}|}{s_i\bar{s}_i}\le \frac{|s_i-\bar{s}_i|}{\bar{s}_i}+\frac{|s_{i+1}-\bar{s}_{i+1}|}{\bar{s}_i}\\
	&\le \frac{2}{\bar{s}_i}\|\bm \alpha-\bar{\bm \alpha}\|\le \frac{2}{|\bar{\alpha}_n|}\|\bm \alpha-\bar{\bm \alpha}\|\le 4\sqrt{n}\|\bm \alpha-\bar{\bm \alpha}\|
	\end{split}
	\]
	And then consider the error of $\frac{\alpha_i\alpha_j}{s_js_{j+1}}$($i>j$):
	\begin{equation}\label{est_H}
	\begin{split}
	&\left|\frac{\alpha_i\alpha_j}{s_js_{j+1}}-\frac{\bar{\alpha}_i\bar{\alpha}_j}{\bar{s}_j\bar{s}_{j+1}}\right|=\frac{|\alpha_i\alpha_j\bar{s}_j\bar{s}_{j+1}
		-\bar{\alpha}_i\bar{\alpha}_js_js_{j+1}|}{s_js_{j+1}\bar{s}_j\bar{s}_{+1}} \\
	&\le\frac{1}{s_js_{j+1}\bar{s}_j\bar{s}_{j+1}}(|\alpha_i\alpha_j\bar{s}_j\bar{s}_{j+1}-\alpha_i\alpha_js_j\bar{s}_{j+1}|+|\alpha_i\alpha_js_j\bar{s}_{j+1}-\alpha_i\alpha_js_js_{j+1}|\\
	&+|\alpha_i\alpha_js_js_{j+1}-\bar{\alpha}_i\alpha_js_js_{j+1}|+|\bar{\alpha}_i\alpha_js_js_{j+1}-\bar{\alpha}_i\bar{\alpha}_js_js_{j+1}|)\\
	&=\frac{\alpha_i\alpha_j\bar{s}_{j+1}}{s_js_{j+1}\bar{s}_j\bar{s}_{j+1}}|\bar{s}_j-s_j|+\frac{\alpha_i\alpha_js_j}{s_js_{j+1}\bar{s}_j\bar{s}_{j+1}}|\bar{s}_{j+1}-s_{j+1}|\\
	&+\frac{\alpha_js_js_{j+1}}{s_js_{j+1}\bar{s}_j\bar{s}_{j+1}}|\alpha_i-\bar{\alpha}_i|+\frac{\bar{\alpha}_is_js_{j+1}}{s_js_{j+1}\bar{s}_j\bar{s}_{j+1}}|\alpha_j-\bar{\alpha}_j|\\
	&\le\frac{|\bar{s}_j-s_j|}{\bar{s}_j}+\frac{|\alpha_j|}{\bar{s}_j}\frac{|\bar{s}_{j+1}-s_{j+1}|}{\bar{s}_{j+1}}+\frac{|\alpha_j|}{\bar{s}_j}\frac{|\alpha_i-\bar{\alpha}_i|}{\bar{s}_{j+1}}
	+\frac{|\alpha_j-\bar{\alpha}_j|}{\bar{s}_j}
	\end{split}
	\end{equation}
	
	We need to estimate $\frac{|\alpha_j|}{\bar{s}_j}$. First, if $|\alpha_j|\le |\bar{\alpha}_j|$, then it holds that $\frac{|\alpha_j|}{\bar{s}_j}\le 1$. When $|\alpha_j|>|\bar{\alpha}_j|$, it follows that
	\[
	\bar{s}_j^2=\bar{\alpha}_j^2+\cdots+\bar{\alpha}_n^2=\alpha_j^2+2\Delta\alpha_j\alpha_j+\Delta\alpha_j^2+\bar{\alpha}_{j+1}^2+\cdots+\bar{\alpha}_n^2,\]
	so we have
	\[\bar{s}_j^2-\alpha_j^2\ge \sum_{k=j+1}^n\bar{\alpha}_k^2-2|\Delta\alpha_j||\alpha_j|\ge\bar{\alpha}_n^2-2|\Delta\alpha_j|.
	\]
	Note that $|\bar{\alpha}_n|>\frac{1}{2\sqrt{n}}$ and $|\Delta\alpha_j|<\frac{1}{8n}$ when $\|\bm \alpha-\bar{\bm\alpha}\|<\frac{1}{8n}$, which indicate  $\bar{s}_j^2-\alpha_j^2\ge \bar{\alpha}_n^2-2|\Delta\alpha_j|>\frac{1}{4n}-\frac{2}{8n}=0$. So it is proved that
	\begin{equation}\label{eq:ajsj}
	\frac{|\alpha_j|}{\bar{s}_j}\le 1
	\end{equation} when $\|\bm \alpha-\bar{\bm\alpha}\|<\frac{1}{8n}$.
	Applying \eqref{eq:ajsj} to \eqref{est_H} gives
	\[
	\begin{split}
	\left|\frac{\alpha_i\alpha_j}{s_js_{j+1}}-\frac{\bar{\alpha}_i\bar{\alpha}_j}{\bar{s}_j\bar{s}_{j+1}}\right|&\le\frac{|\bar{s}_j-s_j|}{\bar{s}_j}+\frac{|\alpha_j|}{\bar{s}_j}\frac{|\bar{s}_{j+1}-s_{j+1}|}{\bar{s}_{j+1}}+\frac{|\alpha_j|}{\bar{s}_j}\frac{|\alpha_i-\bar{\alpha}_i|}{\bar{s}_{j+1}}
	+\frac{|\alpha_j-\bar{\alpha}_j|}{\bar{s}_j}\\
	&\le\frac{|\bar{s}_j-s_j|}{\bar{s}_j}+\frac{|\bar{s}_{j+1}-s_{j+1}|}{\bar{s}_{j+1}}+\frac{|\alpha_i-\bar{\alpha}_i|}{\bar{s}_{j+1}}
	+\frac{|\alpha_j-\bar{\alpha}_j|}{\bar{s}_j}\\
	&\le\frac{4}{\frac{1}{2\sqrt{n}}}\|\bm \alpha-\bar{\bm \alpha}\|=8\sqrt{n}\|\bm \alpha-\bar{\bm \alpha}\|.
	\end{split}
	\]
With the assumption of $\|\bm\alpha-\bar{\bm\alpha}\|<\frac{1}{8n}$, it follows that
	\[
	\begin{split}
	\|\bm{H}_\alpha-\bm{H}_{\bar{\alpha}}\|_F\le 8\sqrt{n}\sqrt{\frac{n(n-1)}{2}+(n-1)}\|\bm\alpha-\bar{\bm \alpha}\|\le 8n^{3/2}\|\bm\alpha-\bar{\bm \alpha}\|.
	\end{split}
	\]
	The proof is finished.
\end{proof}

Now we construct the input $\overline{\bm H}_{{\alpha}}$ of $\PSLQ_{\varepsilon}$ from empirical data $\bar{\bm{\alpha}}$. In this paper, we restrict ourselves under exact arithmetic, i.e., we take $\overline{\bm H}_{\alpha}={\bm H}_{\bar{\alpha}}$. Applying this to Theorem \ref{thm:forward}  yields the following particular error control strategy.

\begin{theorem}\label{thm:main}
	Let $\bm \alpha\in\real^n$ be a unit vector with $|\alpha_n|=\max_i\{|\alpha_i|\}$  and $\varepsilon>0$. Suppose $\bm \alpha$ has an integer relation with $2$-norm bounded from above by $M$. Given empirical data $\bm{\bar{\alpha}}$ with
	\[
	\norm{\bm\alpha-\bar{\bm\alpha}}<\varepsilon_1 < \frac{\varepsilon}{16M C n^{3/2}},
	\]
	if $\PSLQ_\varepsilon$ with
	\[
	\varepsilon_2<\frac{\varepsilon}{2C\alpha_n}
	\]
	returns $\bm m$ with $\norm{\bm m}<M$, then $\abs{\inner{\bm \alpha,\bm m}}<\varepsilon$, where $C = \frac{2(\sqrt{(n-2)\alpha_n^2+1}+\alpha_n)}{\alpha_n}$ and $M>0$.
\end{theorem}

\begin{proof}
	From Lemma \ref{lem:Hbar}, it holds that
	\[
	\norm{\overline{\bm H}_\alpha-{\bm H}_\alpha}_F=\norm{{\bm H}_{\bar{\alpha}}-{\bm H}_\alpha}_F<\frac{\varepsilon}{2M\cdot C}.
	\]
	Then Theorem \ref{thm:forward} implies
	\begin{equation}\label{eq:eps}
	\abs{\inner{\bm \alpha,\bm m}}<C\left(M\frac{\varepsilon}{2M\cdot C} +\alpha_n\varepsilon_2\right)< \frac{\varepsilon}{2}+\frac{\varepsilon}{2}=\varepsilon.
	\end{equation}
	The theorem is proved.
\end{proof}

\subsection{Some Remarks}

It should be noted that the results presented in Theorem \ref{lem:zn-1}, \ref{thm:forward} and \ref{thm:main} can be applied not only to the standard PSLQ algorithm, but also to the multi-pair variant of PSLQ \cite[Section 6]{BaileyBroadhurst2000}. The reason is that all the proofs of these theorems are independent of the swap strategy. The multi-pair variant can be seen as a parallel version of PSLQ, in which several pairs of rows of the matrix $\bm H$ are swaped simultaneously, and it is much more efficient than the standard PSLQ and hence utilized in almost all of the applications in practice.   We also note that the iteration bound in Theorem \ref{thm:term} may not hold for the multi-pair variant; we refer to \cite[page 1729]{BaileyBroadhurst2000} and \cite[Section 3]{FengChenWu2014} for this topic.

It is not difficult to verify that Theorem \ref{thm:main} still holds for $\varepsilon_1<\frac{\omega\varepsilon}{8MCn^{3/2}}$ and $\varepsilon_2<\frac{(1-\omega)\varepsilon}{C\alpha_n}$ for any $0<\omega<1$. The error control strategy given in Theorem \ref{thm:main} just simply takes $\omega=1/2$. Examples in the  next section show the effectiveness of this strategy, but, the optimal choice for $\omega$ is beyond the scope of this paper.

Figure \ref{fig:pic} shows the relationships among the main notations of this paper.
In this figure, the solid lines indicate the routine of $\PSLQ_\varepsilon$ for empirical input data $\bm{\bar{\alpha}}$ with $\norm{\bm \alpha-\bar{\bm\alpha}}< \varepsilon_1$. According to Theorem \ref{thm:main}, if the returned $\bm m$ by $\PSLQ_\varepsilon$ satisfies $\norm{\bm m}<M$ then we can guarantee that $\abs{\inner{\bm m, \bm \alpha}}<\varepsilon$.
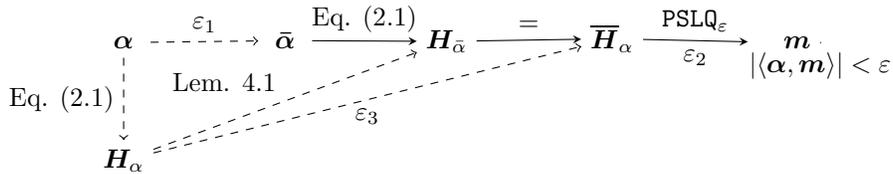
\begin{figure}[!h]
	\centering
	\begin{tikzpicture}
	\matrix (m) [matrix of math nodes,row sep=3em,column sep=4em,minimum width=2em]
	{
		\bm \alpha &\bm{\bar{\alpha}}& \bm H_{\bar{\alpha}} &\overline{\bm H}_{\alpha} &\,\,\,\,\,\,\bm m \\
		\bm H_\alpha & \\};
	\path[-stealth]
	(m-1-1) edge[dashed, ->] node [left] {Eq. \eqref{Halpha}} (m-2-1)
	edge [dashed, ->] node [above] {$\varepsilon_1$} (m-1-2)
	(m-1-2)		edge  node [above] {Eq. \eqref{Halpha}} (m-1-3)
	(m-1-3)		edge  node [above] { = }  (m-1-4)
	(m-1-4)		edge  node [above] {$\PSLQ_\varepsilon$} node [below] {$\varepsilon_2$} (m-1-5)
	(m-2-1) edge [dashed, ->] node [above left] {Lem. \ref{lem:Hbar}} (m-1-3)
	(m-2-1) edge [dashed, ->] node [below] {$\varepsilon_3$} (m-1-4)
	(m-1-5) edge [dashed, -] node [below] {$\abs{\inner{\bm \alpha,\bm m}}<\varepsilon$} (m-1-5);
	\end{tikzpicture}
	\caption{An illustrative picture of relationships among the main notations}
	\label{fig:pic}
\end{figure}

As mentioned previously, high precision arithmetic  must be used for almost all applications of $\PSLQ$. In practice, Bailey (see, e.g., \cite{Bailey2000}) suggested that if one wishes to recover a relation for an $n$-dimensional vector, with coefficients of maximum size $\log_{10} G$ decimal digits, then the input vector $\bm \alpha$ must be specified to at least $n\log_{10} G$ digits, and one must employ floating-point arithmetic accurate to at least $n\log_{10} G$ digits.
However, there seems no theoretical results about how to decide the precision generally. Theorem \ref{thm:forward} and \ref{thm:main} in this paper can be seen as theoretical sufficient conditions for $\PSLQ$ with empirical input data. We show in the next subsection that these theoretical results indeed give some effective strategies for the input data precision and the termination condition in practice.

\subsection{Numerical Examples}

In this subsection, we give some examples to illustrate our strategy of error control based on Theorem \ref{thm:main}. We use our own implementation of $\PSLQ_\varepsilon$ in Maple, which takes the running precision \texttt{Digits}, a target accuracy $\varepsilon$ and an upper bound on the coefficients of the expected relation $G$ as its input. We use $M=\sqrt{n}G$ as its $2$-norm bound and fix  \texttt{Digits\,:=\,200} and \texttt{Digits\,:=\,600}
for the first two examples and the third example, respectively, so that it is  sufficient to guarantee the correctness  and that it can mimic the exact real arithmetic.

\begin{example}[Transcendental numbers]\label{exmp:trans}
	Equation (69) of \cite{Bailey2016} states that $\bm \beta = (t,1, \ln 2,  \ln^2 2, {\pi}^{2})\in\real^5$ has an integer relation  $\bm m=(1, -5, 4, -16, 1)$, where
	\[
	t=\int_{0}^{1}\int_{0}^{1}\left(\frac{x-1}{x+1}\right)^2\left(\frac{y-1}{y+1}\right)^2\left(\frac{xy-1}{xy+1}\right)^2dxdy.
	\]
	We try to recover this relation for $\bm{\alpha}=\bm{\beta}/\norm{\bm{\beta}}$. 	
\end{example}

\begin{figure}[tb]
	\centering
	\begin{tikzpicture}[y=.3cm, x=.6cm]
	
	\draw (0,0) -- coordinate (x axis mid) (10,0);
	\draw (0,0) -- coordinate (y axis mid) (0,20);
	
	\foreach \x in {0,1,...,10}
	\draw (\x,1pt) -- (\x,-3pt)
	node[anchor=north] {\x};
	\foreach \y in {0,2,...,20}
	\draw (1pt,\y) -- (-3pt,\y)
	node[anchor=east] {\y};
	
	\node[below=0.8cm] at (x axis mid) {$\lceil-\log_{10}\varepsilon\rceil$};
	\node[rotate=90, above=0.8cm] at (y axis mid) {$\lceil-\log_{10}y\rceil$};
	
	\draw plot[mark=*, mark options={fill=white}]
	file {F11.data};
	\draw plot[mark=triangle*, mark options={fill=white}]
	file {F21.data};
	
	\draw plot[mark=*]
	file {F12.data};
	\draw plot[mark=triangle*]
	file {F22.data};
	
	\begin{scope}[shift={(1,16)}]
	\draw[yshift=1\baselineskip] (0,0) --
	plot[mark=*, mark options={fill=white}] (0.25,0) -- (0.5,0)
	node[right]{$y=\varepsilon_1$, incorrect output};
	\draw[yshift=0\baselineskip] (0,0) --
	plot[mark=triangle*, mark options={fill=white}] (0.25,0) -- (0.5,0)
	node[right]{$y=\varepsilon_2$, incorrect output};
	\draw[yshift=2\baselineskip] (0,0) --
	plot[mark=triangle*] (0.25,0) -- (0.5,0)
	node[right]{$y=\varepsilon_2$, correct output};
	\draw[yshift=3\baselineskip] (0,0) --
	plot[mark=*, mark options={fill=black}] (0.25,0) -- (0.5,0)
	node[right]{$y=\varepsilon_1$, correct output};
	\end{scope}
	\end{tikzpicture}
	\caption{Error control strategy for Example {\ref{exmp:trans}}}
	\label{fig:exmp1}
\end{figure}
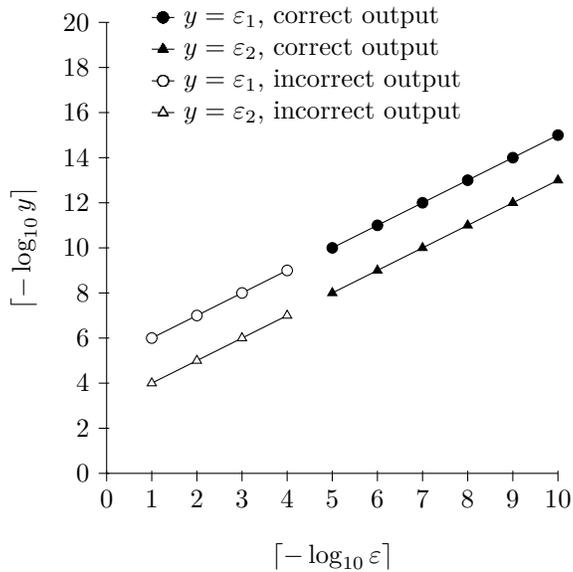

Because of involving transcendental numbers, we can only obtain empirical data of $\bm \alpha$. Suppose  that the maximum of the coefficients is bounded by $G=16$ and that the gap bound for this example is $10^{-6}$. (In fact, by exhaustive search, we can obtain a gap bound that is about $6.37\times 10^{-6}$.) Thus, the target precision $\varepsilon$ is set as $\varepsilon=10^{-5}$. It means that we want to find an integer vector $\bm m$ such that $\abs{\inner{\bm \alpha,\bm m}}<\varepsilon=10^{-5}$.  According to Theorem \ref{thm:main}, we obtain that $\varepsilon_1\approx 2.60\times 10^{-11}$ and $\varepsilon_2\approx 8.39\times 10^{-8}$. We run this example in the computer algebra system \texttt{Maple}. After $30$ iterations of PSLQ, the procedure returns a relation $\bm m=(1, -5, 4, -16, 1)$, which is an exact integer relation for $\bm \alpha$.

If we do not know a gap bound on $\abs{\inner{\bm m,\bm\alpha}}$, we can test $\varepsilon = 10^{-i}$ for $i=1,2,\cdots,10$, where the corresponding $\varepsilon_1$ and $\varepsilon_2$ are decided according to Theorem \ref{thm:main}.  As shown in Figure \ref{fig:exmp1},
for $i=1,2,3,4$, no correct answer is obtained, but for $5\le i\le 10$ the procedure always returns the same relation $\bm m$. Further, the difference between  $\lceil-\log_{10}\varepsilon_1\rceil$ and $\lceil-\log_{10}\varepsilon_2\rceil$ does not change for different $\varepsilon$.

Bailey's estimation  is $\lceil n\log_{10}G\rceil = 7$ decimal digits that indicates $\varepsilon_1<10^{-7}$, which is relatively compact for the above setting. However, Bailey's estimation still has the following drawbacks. For one thing, Bailey's estimation does not suggest when the algorithm terminates, i.e., how to choose $\varepsilon_2$, while Theorem \ref{thm:main} suggests the quantity that  $\varepsilon_2$ should be larger than $\varepsilon_1$. This is consistent with intuition: the error would be amplified by exact computation with empirical data as input. In fact, if we do not have the error control strategy as indicated by Theorem \ref{thm:main}, we can only use a trial-and-error approach to decide the termination precision $\varepsilon_2$, since the procedure may miss the correct answer for an incorrect $\varepsilon_2$, even with relatively high precision.

For another thing, if we do not know such a tight bound on the maximum coefficient of the relation, instead, for example, we only know $G\le 10^5$. For the same $\varepsilon$, we now have $\varepsilon_1\approx 4.16\times 10^{-15}$ and $\varepsilon_2\approx 8.39\times 10^{-8}$, for which our procedure work correctly, while at least  $\lceil n\log_{10}G\rceil = 25$ decimal digits is needed according to Bailey's estimation, which implies $\varepsilon_1\le 10^{-25}$.  For this example, by Bailey's estimation, $\lceil-\log_{10}\varepsilon_1\rceil$ increases linearly with $\lceil\log_{10}G\rceil$, whose slope is $n=5$. According to Theorem \ref{thm:main}, $\lceil-\log_{10}\varepsilon_1\rceil$ also increases linearly with $\lceil\log_{10}G\rceil$, but the slope is about $1$ only. In fact, according to Theorem \ref{thm:main}, we have $\lceil-\log_{10}\varepsilon_1\rceil\ge \lceil\log G +\log_{10}(16n^{5/2}C) -\log_{10}\varepsilon\rceil$.

\begin{example}[Algebraic numbers]\label{exmp:alg}
	Let $\alpha = (\sqrt [5]{3}+\sqrt [4]{2})^{-1}$ and let $\bm \alpha$ be the normalized vector of $(\alpha^{20},\alpha^{19},\cdots,\alpha,1)$. In this example, we try to recover the coefficients of
	the minimal polynomial of $\alpha$. Suppose that we know in advance that the $\infty$-norm of the integer relation is at most $G=7440$.
\end{example}

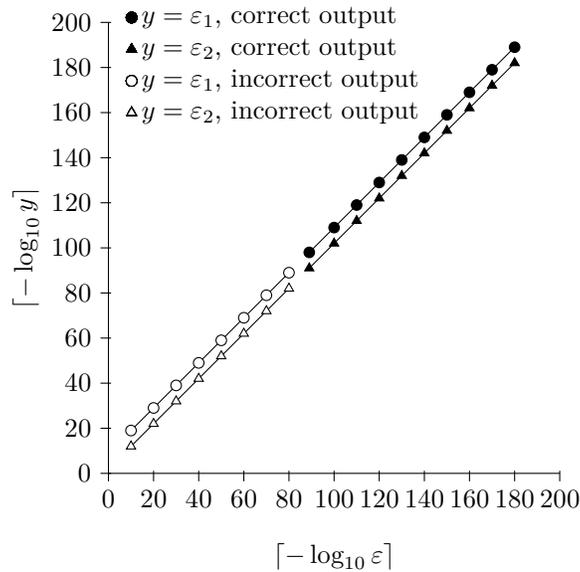
\begin{figure}[!h]
	\centering
	\begin{tikzpicture}[y=.03cm, x=.03cm]
	
	\draw (0,0) -- coordinate (x axis mid) (200,0);
	\draw (0,0) -- coordinate (y axis mid) (0,200);
	
	\foreach \x in {0,20,...,200}
	\draw (\x,1pt) -- (\x,-3pt)
	node[anchor=north] {\x};
	\foreach \y in {0,20,...,200}
	\draw (1pt,\y) -- (-3pt,\y)
	node[anchor=east] {\y};
	
	\node[below=0.8cm] at (x axis mid) {$\lceil-\log_{10}\varepsilon\rceil$};
	\node[rotate=90, above=0.8cm] at (y axis mid) {$\lceil-\log_{10}y\rceil$};
	
	\draw plot[mark=*, mark options={fill=white}]
	file {e11.data};
	\draw plot[mark=*]
	file {e12.data};
	\draw plot[mark=triangle*, mark options={fill=white}]
	file {e21.data};
	\draw plot[mark=triangle*]
	file {e22.data};

	\begin{scope}[shift={(10,160)}]
	\draw[yshift=1\baselineskip] (0,0) --
	plot[mark=*, mark options={fill=white}] (0.25,0) -- (0.5,0)
	node[right]{$y=\varepsilon_1$, incorrect output};
	\draw[yshift=0\baselineskip] (0,0) --
	plot[mark=triangle*, mark options={fill=white}] (0.25,0) -- (0.5,0)
	node[right]{$y=\varepsilon_2$, incorrect output};
	\draw[yshift=2\baselineskip] (0,0) --
	plot[mark=triangle*] (0.25,0) -- (0.5,0)
	node[right]{$y=\varepsilon_2$, correct output};
	\draw[yshift=3\baselineskip] (0,0) --
	plot[mark=*, mark options={fill=black}] (0.25,0) -- (0.5,0)
	node[right]{$y=\varepsilon_1$, correct output};
	\end{scope}
	\end{tikzpicture}
	\caption{Error control strategy for Example {\ref{exmp:alg}}}
	\label{fig:exmp2}
\end{figure}

Bailey's estimation  suggests that $\bm \alpha$ should be computed with  at least $\lceil n\log_{10} G\rceil = 82$ exact decimal digits, which implies  $\varepsilon_1<10^{-82}$. However, $\PSLQ_\varepsilon$ does not  return a relation with coefficient bounded by $7440$. This may be caused by the fact that Bailey's estimation is not sufficient to compute an integer relation.

Let us set $\varepsilon = 10^{-89}$ so that $\varepsilon_1\approx 1.73\times10^{-98}$ and $\varepsilon_2\approx 4.99\times10^{-91}$, and our procedure returns a relation
\[\begin{split}
\bm m = (49, -1080, 3960, -3360, 80, -108, -6120, -7440,& \\
-80, 0, 54, -1560, 40, 0, 0, -12, -10, 0, 0, 0, 1&)
\end{split}
\]
after $3525$ iterations. It can be checked that this relation  corresponds exactly to the coefficients of the minimal polynomial of  $\alpha$.

For the same $\varepsilon$ and $\varepsilon_1$, if we do not set $\varepsilon_2$ as suggested by Theorem \ref{thm:main}, say, $\varepsilon_2\approx 10^{-96}$, then the procedure misses the correct relation.

If we set $\varepsilon=10^{-88}$, our procedure does not return the correct answer. This can be seen as evidence for that the sharp gap bound is near to $10^{-89}$. We also test for $\varepsilon=10^{-(100-10i)}$ with $i=1, 2, \cdots, 9$. Each of these tests does not return the correct answer. If we set $\varepsilon$ more strictly, which means paying more precision, for example $\varepsilon=10^{-(100+10i)}$ with $i=1, 2, \cdots, 8$, the procedure always works well and returns the same $\bm m$ as above. The quantities $\lceil-\log_{10}\varepsilon_1\rceil$ and $\lceil-\log_{10}\varepsilon_2\rceil$ obtained from Theorem \ref{thm:main} are as shown in Figure \ref{fig:exmp2}.

\begin{example}[Algebraic numbers with higher degree]\label{exmp:alg-high}
	Let $\alpha = (\sqrt [7]{3}+\sqrt [7]{2})^{-1}$ and let $\bm \alpha$ be the normalized vector of $(\alpha^{49},\alpha^{48},\cdots,\alpha,1)$.
\end{example}

For this example, the dimension is $50$ and the $\infty$-norm of the integer relation is $G=966420105$. Bailey's estimation  suggests that $\bm \alpha$ should be computed with at least  $\lceil n\log_{10} G\rceil = 450$ exact decimal digits, which implies $\varepsilon_1<10^{-450}$. Under this setting, $\PSLQ_\varepsilon$ fails to find the correct relation. The reason is that this precision is not enough to achieve the gap bound.
In fact, according to our tests, the gap bound for this example is about $10^{-487}$; see Figure \ref{fig:exmp3}. This shows that Bailey's estimation is not sufficient, but still necessary.

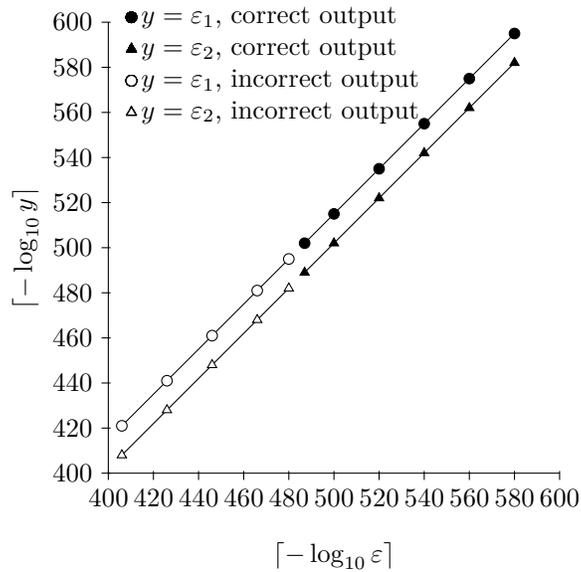
\begin{figure}[!h]
	\centering
	\begin{tikzpicture}[y=.03cm, x=.03cm]
	
	\draw (400,400) -- coordinate (x axis mid) (400,600);
	\draw (400,400) -- coordinate (y axis mid) (600,400);
	
	\foreach \x in {400,420,...,600}
	\draw (\x,401) -- (\x,397)
	node[anchor=north] {\x};
	\foreach \y in {400,420,...,600}
	\draw (401,\y) -- (397,\y)
	node[anchor=east] {\y};
	
	\node[below=0.8cm] at (500,400) {$\lceil-\log_{10}\varepsilon\rceil$};
	\node[rotate=90, above=0.8cm] at (400, 500) {$\lceil-\log_{10}y\rceil$};
	
	\draw plot[mark=*, mark options={fill=white}]
	file {H11.data};
	\draw plot[mark=*]
	file {H21.data};
	\draw plot[mark=triangle*, mark options={fill=white}]
	file {H12.data};
	\draw plot[mark=triangle*]
	file {H22.data};

	\begin{scope}[shift={(410,560)}]
	\draw[yshift=1\baselineskip] (0,0) --
	plot[mark=*, mark options={fill=white}] (0.25,0) -- (0.5,0)
	node[right]{$y=\varepsilon_1$, incorrect output};
	\draw[yshift=0\baselineskip] (0,0) --
	plot[mark=triangle*, mark options={fill=white}] (0.25,0) -- (0.5,0)
	node[right]{$y=\varepsilon_2$, incorrect output};
	\draw[yshift=2\baselineskip] (0,0) --
	plot[mark=triangle*] (0.25,0) -- (0.5,0)
	node[right]{$y=\varepsilon_2$, correct output};
	\draw[yshift=3\baselineskip] (0,0) --
	plot[mark=*, mark options={fill=black}] (0.25,0) -- (0.5,0)
	node[right]{$y=\varepsilon_1$, correct output};
	\end{scope}
	\end{tikzpicture}
	\caption{Error control strategy for Example {\ref{exmp:alg-high}}}
	\label{fig:exmp3}
\end{figure}


When we set $\varepsilon = 10^{-487}$ so that $\varepsilon_1\approx 1.61\times10^{-502}(<10^{-450})$ and $\varepsilon_2\approx 3.47\times10^{-489}$ according to Theorem \ref{thm:main}, then  our procedure returns the correct relation corresponding to the coefficients of the minimal polynomial of  $\alpha$ after $45385$ iterations. Furthermore, the similar phenomenon showed in Figure \ref{fig:exmp1} and \ref{fig:exmp2} also appears for this example. When we set $\varepsilon$ smaller than $10^{-487}$ (and set $\varepsilon_1$ and $\varepsilon_2$ accordingly), $\PSLQ_\varepsilon$ always returns the same integer relation, as shown in Figure \ref{fig:exmp3}. This shows that our error control strategy given in Theorem \ref{thm:main}  plays an important role for the correctness of $\PSLQ_\varepsilon$.

From examples above, we have the following two observations. Firstly, if one does not decide $\varepsilon_1$ and $\varepsilon_2$ by the error control strategy in Theorem \ref{thm:main}, then one may miss the correct relation. Secondly, with an effective $\varepsilon$, we  always obtain the same relation if we use the error control strategy  in Theorem \ref{thm:main}. This observation may be taken as strong evidence  that the returned relation is a true integer relation. In fact, assume that for all arbitrary small $\varepsilon>0$, $\PSLQ_\varepsilon$  always returns the same relation. Then the relation must be an exact integer relation in the sense that $\PSLQ_\varepsilon\rightarrow\bm m$ for $\varepsilon\rightarrow 0$. However, if no gap bound is known, determining  whether the returned relation is an exact integer relation within finite steps is still open.

\section{Disscussion and Conclusion}

In this paper, we give a new invariant relation of the celebrated integer relation finding  algorithm \PSLQ, and hence introduce a new termination condition for $\PSLQ_{\varepsilon}$. The new termination condition
allows us to compute integer relations by $\PSLQ_{\varepsilon}$ with empirical data as its input. By a perturbation analysis, we disclose the relationship between the accuracy of the input data ($\varepsilon_1$) and the output quality ($\varepsilon$, an upper bound on the absolute value of the inner product of the intrinsic data and the computed relation) of the algorithm. This relationship still holds for the multi-pair variant of PSLQ. Examples show that our error control strategies based on this relationship are very helpful in practice.

We note that all results presented in this paper are under the exact arithmetic computational model. Although we obtain some results about the error control for applications, we did not analyze the algorithm under an inexact arithmetic model, such as floating-point arithmetic. However, we believe that the results in this paper, say Theorem \ref{thm:forward}, would be indispensable in the analysis of a numerical $\PSLQ$ algorithm.

In fact, it is an intriguing topic to design and analyze an efficient numerical $\PSLQ$ algorithm.  For the moment, the main obstacle is to give a reasonable bound on the entries of unimodular matrices produced by the algorithm. Now, we can  only give an upper bound that is double exponential with respect to the working dimension, and hence resulting in an exponential time algorithm. Thus, it is a very interesting challenge to obtain an upper bound similar to, e.g., \cite[Lemma 6]{SaruchiMorelStehleVillard2014}, where the upper bound is of single exponential in the dimension.

\section*{Acknowledgments}
We would like to thank an anonymous referee for helpful suggestions that greatly improved the presentation of this papr.

\bibliographystyle{amsplain}
\newcommand{\noopsort}[1]{}
\providecommand{\bysame}{\leavevmode\hbox to3em{\hrulefill}\thinspace}
\providecommand{\MR}{\relax\ifhmode\unskip\space\fi MR }
\providecommand{\MRhref}[2]{%
  \href{http://www.ams.org/mathscinet-getitem?mr=#1}{#2}
}
\providecommand{\href}[2]{#2}

\appendix
\section{Proof of Lemma \ref{lem:normH}}\label{sec:condnum}

We consider the following  submatrix of $\bm{H}_\alpha$, denoted by $\bm{H}_{[1..n-1]}$,
$$\bm{H}_{[1..n-1]}=\begin{pmatrix}\frac{s_2}{s_1}&0&0&\cdots&0&0\\
\frac{-\alpha_2\alpha_1}{s_1s_2}&\frac{s_3}{s_2}&0&\cdots&0&0 \\
\frac{-\alpha_3\alpha_1}{s_1s_2}&\frac{-\alpha_3\alpha_2}{s_2s_3}&\frac{s_4}{s_3}&\cdots&0&0\\
\frac{-\alpha_4\alpha_1}{s_1s_2}&\frac{-\alpha_4\alpha_2}{s_2s_3}&\frac{-\alpha_4\alpha_3}{s_3s_4}&\cdots&0&0\\
\vdots &\vdots&\vdots&\vdots&\vdots&\vdots\\
\frac{-\alpha_{n-2}\alpha_1}{s_1s_2}&\frac{-\alpha_{n-2}\alpha_2}{s_2s_3}&\frac{-\alpha_{n-2}\alpha_3}{s_3s_4}&\cdots&\frac{s_{n-1}}{s_{n-2}}&0\\
\frac{-\alpha_{n-1}\alpha_1}{s_1s_2}&\frac{-\alpha_{n-1}\alpha_2}{s_2s_3}&\frac{-\alpha_{n-1}\alpha_3}{s_3s_4}&\cdots&\frac{-\alpha_{n-1}\alpha_{n-2}}{s_{n-2}s_{n-1}}&\frac{s_n}{s_{n-1}}
\end{pmatrix}.$$
By linear algebra, its inverse is
\begin{equation}
\bm{H}^{-1}_{[1..n-1]}=\begin{pmatrix}\frac{s_1}{s_2}&0&0&\cdots&0&0\\
\frac{\alpha_1\alpha_2}{s_2s_3}&\frac{s_2}{s_3}&0&\cdots&0&0 \\
\frac{\alpha_1\alpha_3}{s_3s_4}&\frac{\alpha_2\alpha_3}{s_3s_4}&\frac{s_3}{s_4}&\cdots&0&0\\
\frac{\alpha_1\alpha_4}{s_4s_5}&\frac{\alpha_2\alpha_4}{s_4s_5}&\frac{\alpha_3\alpha_4}{s_4s_5}&\cdots&0&0\\
\vdots &\vdots&\vdots&\vdots&\vdots&\vdots\\
\frac{\alpha_1\alpha_{n-2}}{s_{n-2}s_{n-1}}&\frac{\alpha_2\alpha_{n-2}}{s_{n-2}s_{n-1}}&\frac{\alpha_3\alpha_{n-2}}{s_{n-2}s_{n-1}}&\cdots&\frac{s_{n-2}}{s_{n-1}}&0\\
\frac{\alpha_1\alpha_{n-1}}{s_{n-1}s_n}&\frac{\alpha_2\alpha_{n-1}}{s_{n-1}s_n}&\frac{\alpha_3\alpha_{n-1}}{s_{n-1}s_n}&\cdots&\frac{\alpha_{n-2}\alpha_{n-1}}{s_{n-1}s_n}&\frac{s_{n-1}}{s_n}
\end{pmatrix}.
\end{equation}

In the following, we compute the F-norm of $\bm{H}^{-1}_{[1..n-1]}$.
First, consider the $j$-th column of  $\bm{H}_{[1..n-1]}^{-1}$:
\begin{equation*}
\begin{split}
\|H_j^{-1}\|^2&=\frac{s_j^2}{s_{j+1}^2}+\sum_{k=j+1}^{n-1}\frac{\alpha_j^2\alpha_k^2}{s_k^2s_{k+1}^2}=\frac{s_j^2}{s_{j+1}^2}
+\alpha_j^2\sum_{k=j+1}^{n-1}\frac{\alpha_k^2}{s_k^2s_{k+1}^2}\\
&=\frac{s_j^2}{s_{j+1}^2}+\alpha_j^2\sum_{k=j+1}^{n-1}(\frac{1}{s_{k+1}^2}-\frac{1}{s_k^2})=\frac{s_j^2}{s_{j+1}^2}
+\alpha_j^2(\frac{1}{s_n^2}-\frac{1}{s_{j+1}^2})\\
&=\frac{s_j^2-\alpha_j^2}{s_{j+1}^2}+\frac{\alpha_j^2}{s_n^2}=\frac{s_{j+1}^2}{s_{j+1}^2}+\frac{\alpha_j^2}{\alpha_n^2}
=1+\frac{\alpha_j^2}{\alpha_n^2},
\end{split}
\end{equation*}
so we have
\[
\begin{split}
\|\bm{H}_{[1..n-1]}^{-1}\|_F^2&=\sum_{j=1}^{n-1}\|H_j^{-1}\|^2=(n-1)+\frac{\sum_{j=1}^{n-1}\alpha_j^2}{\alpha_n^2}\\
&=(n-1)+\frac{\|\alpha\|^2-\alpha_n^2}{\alpha_n^2}=(n-2)+\frac{\|\alpha\|^2}{\alpha_n^2}.
\end{split}
\]
In addition, we can compute the F-norm of $\bm{H}_{[1..n-1]}$ as follows:
\[
\begin{split}
\|\bm{H}_{[1..n-1]}\|_F^2&=\|\bm{H}_{\alpha}\|_F^2-\sum_{i=1}^{n-1}\frac{\alpha_n^2\alpha_i^2}{s_i^2s_{i+1}^2}=(n-1)-\alpha_n^2\sum_{i=1}^{n-1}\frac{\alpha_i^2}{s_i^2s_{i+1}^2}\\
&=(n-1)-\alpha_n^2\sum_{i=1}^{n-1}(\frac{1}{s_{i+1}^2}-\frac{1}{s_i^2})=(n-1)-\alpha_n^2(\frac{1}{s_n^2}-\frac{1}{s_1^2})\\
&=(n-1)-1+\frac{\alpha_n^2}{\|\alpha\|^2}=(n-2)+\frac{\alpha_n^2}{\|\alpha\|^2},
\end{split}
\]
as claimed in Lemma \ref{lem:normH}.

\section{Proof of Theorem \ref{thm:term}}\label{app:term}

Define the $\Pi$ function after exactly $k$ iterations as follows
\[
\Pi(k) = \prod_{j=1}^{n-1}\max\left(\abs{h_{i,i}(k)}, \frac{h_{\max}(k)}{\gamma^{n-1}}\right)^{n-j},
\]
where $h_{\max}(k)$ is the maximum of $\abs{h_{i,i}(k)}$ for $i=1,2,\cdots, n-1$. It obviously holds that
$$\Pi(k) = \prod_{j=1}^{n-1}\max\left(\abs{h_{i,i}(k)}, \frac{h_{\max}(k)}{\gamma^{n-1}}\right)^{n-j}\ge \left(\frac{h_{\max}(k)}{\gamma^{n-1}}\right)^{\frac{n(n-1)}{2}}.$$

First, we assert that $h_{\max}(k)\ge h_{\max}(k+1)$. Size reduction does not affect $h_{i,i}(k)$, neither does $h_{\max}$. Let us consider the change of $h_{\max}$ in the  Bergman swap. Let Bergman swap occur at the $r$-th row. For the case of $r<n-1$, after the Bergman swap, we have that
\begin{align*}
&|h_{r,r}(k+1)|<\frac{1}{\tau} |h_{r,r}(k)|<|h_{r,r}(k)|=h_{\max}(k)\\
&|h_{r+1,r+1}(k+1)|=\frac{|h_{r,r}(k)h_{r+1,r+1}(k)|}{\sqrt{h_{r+1,r}^2(k)+h_{r+1,r+1}^2(k)}}\le |h_{r,r}(k)|=h_{\max}(k)
\end{align*}
and the others are unchanged, i.,e. $h_{i,i}(k+1)=h_{i,i}(k)$ for $i=1,\cdots,r-1,r+1,\cdots,n-1$. It shows that $h_{\max}(k)\ge h_{\max}(k+1)$ for $r<n-1$. For the case of $r=n-1$, after the Bergman swap, it holds that $|h_{n-1,n-1}(k+1)|<\frac{1}{\rho}|h_{n-1,n-1}(k)|\le h_{\max}(k)$ and the other $h_{i,i}$'s are unchanged. Therefore it is obtained that $h_{\max}(k)\ge h_{\max}(k+1)$ for $r=n-1$.

Second, we show that $ \Pi(k)>\tau\Pi(k+1)$. Let Bergman swap occurs at row $r$.

\noindent Case $r=n-1$: We have
\begin{equation*}
\begin{split}
&\frac{\Pi(k)}{\Pi(k+1)}=\frac{\max\{|h_{n-1,n-1}(k)|,\frac{h_{\max}(k)}{\gamma^{n-1}}\}}{\max\{|h_{n-1,n-1}(k+1)|,\frac{h_{\max}(k+1)}{\gamma^{n-1}}\}}
=\frac{|h_{n-1,n-1}(k)|}{\max\{|h_{n-1,n}(k)|,\frac{h_{\max}(k+1)}{\gamma^{n-1}}\}}\\
&=\begin{cases}
\frac{|h_{n-1,n-1}(k)|}{|h_{n,n-1}(k)}\ge \rho\ge\tau, & \text{when $h_{n,n-1}(k)>\frac{h_{\max}(k+1)}{\gamma^{n-1}}$,} \\
\frac{|h_{n-1,n-1}(k)|}{\frac{h_{\max}(k+1)}{\gamma^{n-1}}}\ge \frac{|h_{n-1,n-1}(k)|}{\frac{h_{\max}(k)}{\gamma^{n-1}}}\ge\gamma\ge\tau,  & \text{otherwise,}
\end{cases}
\end{split}
\end{equation*}
where we used $\gamma^{n-1}h_{n-1,n-1}(k)\ge h_{\max}(k)$ and $h_{n-1,n-1}(k+1)=h_{n-1,n}(k)$.

Cases $r<n-1$: Let $$A=\frac{\max\{|h_{r,r}(k)|,\frac{h_{\max}(k)}{\gamma^{n-1}}\}}{\max\{|h_{r,r}(k+1)|,\frac{h_{\max}(k+1)}{\gamma^{n-1}}\}},\,  B=\frac{\max\{|h_{r+1,r+1}(k)|,\frac{h_{\max}(k)}{\gamma^{n-1}}\}}{\max\{|h_{r+1,r+1}(k+1)|,\frac{h_{\max}(k+1)}{\gamma^{n-1}}\}}.$$
Then $\frac{\Pi(k)}{\Pi(k+1)}=A(AB)^{n-r-1}$. Set $\eta=h_{r,r}(k)$, $\lambda=h_{r+1,r+1}(k)$, $\beta=h_{r+1,r}(k)$ and $\delta=\sqrt{\beta^2+\lambda^2}$. Noticing that $h_{\max}(k)\ge h_{\max}(k+1)$ and $|\eta|>\frac{h_{\max}(k)}{\gamma^{n-1}}$ yields
\begin{equation}\label{A_est}
\begin{split}
A&=\frac{\max\{|h_{r,r}(k)|,\frac{h_{\max}(k)}{\gamma^{n-1}}\}}{\max\{|h_{r,r}(k+1)|,\frac{h_{\max}(k+1)}{\gamma^{n-1}}\}}=
\frac{|\eta|}{\max\{\delta,\frac{h_{\max}(k+1)}{\gamma^{n-1}}\}}\\
&=\begin{cases}
\frac{|\eta|}{\delta}=\frac{1}{\sqrt{\frac{\beta^2}{\eta^2}+\frac{\lambda^2}{\eta^2}}}\ge \tau, &\text{when $\delta\ge \frac{h_{\max}(k+1)}{\gamma^{n-1}}$,}\\
\frac{|\eta|}{\frac{h_{\max}(k+1)}{\gamma^{n-1}}}=\frac{|\eta|\gamma^{n-1}}{h_{\max}(k+1)}\ge \frac{|\eta|\gamma^{n-1}}{h_{\max}(k)}\ge\gamma\ge\tau,  &\text{otherwise}.
\end{cases}
\end{split}
\end{equation}
And then, we consider $AB=A\cdot \frac{\max\{|\lambda|,\frac{h_{\max}(k)}{\gamma^{n-1}}\}}{\max\{\frac{|\eta\lambda|}{\delta},\frac{h_{\max}(k+1)}{\gamma^{n-1}}\}}$. When  $|\lambda|\ge \frac{h_{\max}(k)}{\gamma^{n-1}}$, it is easily deduced that $\delta\ge|\lambda|\ge \frac{h_{\max}(k)}{\gamma^{n-1}}\ge \frac{h_{\max}(k+1)}{\gamma^{n-1}}$ and $\frac{|\eta\lambda|}{\delta}>\lambda\ge \frac{h_{\max}(k+1)}{\gamma^{n-1}}$. Hence from equation (\ref{A_est}) it holds that
$$AB=A\cdot\frac{|\lambda|}{\frac{|\eta\lambda|}{\delta}}=A\cdot\frac{\delta}{|\eta|}
=\frac{|\eta|}{\delta}\cdot\frac{\delta}{|\eta|}=1.$$
When $|\lambda|< \frac{h_{\max}(k)}{\gamma^{n-1}}$, it holds that
\begin{equation*}
\begin{split}
AB=&A\cdot\frac{\frac{h_{\max}(k)}{\gamma^{n-1}}}{\max\{\frac{|\eta\lambda|}{\delta},\frac{h_{\max}(k+1)}{\gamma^{n-1}}\}}\\
=&\begin{cases}
A\cdot\frac{\frac{h_{\max}(k)}{\gamma^{n-1}}}{\frac{h_{\max}(k+1)}{\gamma^{n-1}}}\ge A\ge\tau> 1, \text{   if $\frac{|\eta\lambda|}{\delta}\le \frac{h_{\max}(k+1)}{\gamma^{n-1}} $,} \\
A\cdot\frac{\frac{h_{\max}(k)}{\gamma^{n-1}}}{\frac{|\eta\lambda|}{\delta}}=\begin{cases}
\frac{|\eta|}{\delta}\cdot \frac{h_{\max}(k)}{\gamma^{n-1}}\cdot\frac{\delta}{|\eta\lambda|}=\frac{h_{\max}(k)}{\lambda\gamma^{n-1}}>1, &\text{else if $\delta>\frac{h_{\max}(k+1)}{\gamma^{n-1}}$,} \\
\frac{|\eta|}{\frac{h_{\max}(k+1)}{\gamma^{n-1}}}\cdot\frac{\frac{h_{\max}(k)}{\gamma^{n-1}}}{\frac{|\eta\lambda|}{\delta}}\ge\frac{\delta}{|\lambda|}\ge 1, &\text{otherwise.}
\end{cases}
\end{cases}
\end{split}
\end{equation*}
Up to now, we have shown that $AB\ge 1$. Therefore
\begin{equation*}
\frac{\Pi(k)}{\Pi(k+1)}=A[AB]^{n-r-1}>A>\tau.
\end{equation*}
It is proved that
\begin{equation}\label{eq:alg2}
\left(\frac{h_{\max}(k)}{\gamma^{n-1}}\right)^{\frac{n(n-1)}{2}} \le\Pi(k)\le\frac{1}{\tau^k}.
\end{equation}
From $\tau>1$, we have
\[
k\le \frac{n(n-1)((n-1)\log\gamma + \log\frac{1}{h_{\max}(k)})}{2\log\tau}.
\]
From $\abs{h_{n, n-1}(k)}< \abs{h_{n-1, n-1}(k)} < h_{\max}(k)$, it always holds that  $h_{\max}(k)\ge \varepsilon_2$ before termination. Hence, we deduce that
\[
k\le \frac{n(n-1)[(n-1)\log\gamma + \log\frac{1}{\varepsilon_2}]}{2\log\tau},
\]
which completes the proof.

\end{document}